\documentclass[10pt]{amsart}
\UseRawInputEncoding
\usepackage{amsfonts}
\usepackage{amssymb}
\usepackage{amsmath,amsthm}
\usepackage{amscd}
\usepackage{graphics}
\usepackage{amsmath,amssymb}
\usepackage{cancel}
\usepackage{pdfsync}
\usepackage{mathrsfs}
\usepackage{wasysym}
\usepackage{latexsym}
\newcommand{\rd}{\mathrm{d}}
\newcommand{\beq}{\begin{equation}}
\newcommand{\eeq}{\end{equation}}
\newcommand{\bea}{\begin{eqnarray}}
\newcommand{\eea}{\end{eqnarray}}
\newcommand{\nn}{\nonumber}

\newcommand{\bk}{\begin{cases}}
\newcommand{\ek}{\end{cases}}

\newtheorem{theorem}{Theorem}

\newtheorem{corollary}[theorem]{Corollary}

\newtheorem{definition}[theorem]{Definition}
\newtheorem{example}[theorem]{Example}

\newtheorem{lemma}[theorem]{Lemma}

\newtheorem{proposition}[theorem]{Proposition}
\newtheorem{remark}[theorem]{Remark}

\DeclareMathOperator{\rank}{rank}

\begin{document}

\title[Haantjes Algebras and Diagonalization]{Haantjes Algebras and Diagonalization}

\author{Piergiulio Tempesta}
\address{Departamento de F\'{\i}sica Te\'{o}rica II, Facultad de F\'{\i}sicas, Universidad
Complutense, 28040 -- Madrid, Spain
 and Instituto de Ciencias Matem\'aticas, C/ Nicol\'as Cabrera, No 13--15, 28049 Madrid, Spain.}
\email{p.tempesta@fis.ucm.es, piergiulio.tempesta@icmat.es}
\author{Giorgio Tondo}
\address{Dipartimento di Matematica e Geoscienze, Universit\`a  degli Studi di Trieste,
piaz.le Europa 1, I--34127 Trieste, Italy.}
\email{tondo@units.it}
\date{November 05, 2020}



\maketitle

\begin{abstract}
We introduce the notion of Haantjes algebra: It consists of an assignment of a family of operator fields on a differentiable manifold,
each of them with vanishing Haantjes torsion. They are also required to satisfy suitable compatibility conditions. Haantjes algebras naturally generalize several known interesting geometric structures, arising in Riemannian geometry and in the theory of integrable systems. At the same time, as we will show, they play a crucial role in the theory of diagonalization of operators on differentiable manifolds.
Assuming that the operators of a Haantjes algebra are semisimple and commute, we shall prove that there exists a set of local coordinates where all operators can be diagonalized simultaneously.
Moreover, in the general, non-semisimple case, they acquire simultaneously, in a suitable local chart, a block-diagonal form.
\end{abstract}




\tableofcontents

\section{Introduction}

The purpose of this paper is to introduce a new geometric--algebraic structure,  based on the notion of Haantjes torsion, that we shall call \textit{Haantjes algebra}.
The Haantjes torsion was introduced in 1955 by J. Haantjes in~\cite{Haa}, as a natural generalization of the torsion defined by Nijenhuis in~\cite{Nij}.
The theory of  tensor fields with vanishing Nijenhuis torsion has been intensively investigated in the last forty years, mainly due  to its relevance in the theory  of almost complex structures~\cite{NN} and its applications to the theory of integrable systems and separation of variables, where they
are usually called recursion operators~\cite{GVY,DeF,MV,IMM,FP,BKM}. However, quite surprisingly, the relevance of Haantjes's differential--geometric work  has not been recognized for a long time, except for  some  notable applications to Hamiltonian
systems of hydrodynamic type~\cite{Mokhov,BogJMP,BogJMP2,FeMa, GOS}.
For a nice review of  classical and more recent results regarding the theory of Nijenhuis and Haantjes tensors,  see Ref.~\cite{YKS}.

Our work is inspired, one the one hand, by the construction of Haantjes manifolds proposed in~\cite{MFrob,MGall13,MGall17}; on the other hand, by the concept  of  $\omega \mathcal{H}$ (or symplectic-Haantjes) manifolds that we
have recently introduced in~\cite{TT2016prepr} in connection with the theory of classical integrable
systems. In this context, $\omega \mathcal{H}$ structures provide us with a natural theoretical framework for dealing with the integrability and
separability properties of Hamiltonian systems, which  parallels and completes  the approach offered by the Nijenhuis geometry.

In this article, we will extend the previous ideas by proposing the general and abstract notion of Haantjes algebras. They consist essentially of a differentiable manifold $M$  endowed with a family of endomorphisms of the tangent bundle
with vanishing Haantjes torsion, which are compatible with each other. In this more general framework, the existence of an underlying symplectic structure is no longer required. In this work, we shall focus on operators with \textit{real} eigenvalues only.

A Haantjes algebra is a flexible tool that can be specialized to treat many different interesting constructions in a natural and unified language.
For example, Magri's Haantjes manifolds are a specially relevant instance of Haantjes algebras. Another important class of Haantjes algebras is represented by
 the \emph{Killing--St\"ackel algebras} introduced in~\cite{BCR02} on a Riemannian manifold, with the aim to characterize separation of variables in classical Hamiltonian systems.

However,  apart from the intrinsic interest of  structures combining the Haantjes geometry with  symplectic or Riemannian geometry, our main motivation is the abstract problem of the
\textit{diagonalization} of operators on a differentiable manifold. Indeed, we shall prove that the algebras of Haantjes  operators introduced in Section 4 can be diagonalized simultaneously
in a suitable local coordinate system, called a \textit{Haantjes chart}. Note that no hermiticity assumption is needed: we only require that the Haantjes operators are a set of pointwise diagonalizable commuting operators.
Since the vanishing of the Haantjes torsion of an operator is a fourth-degree requirement in its components (see formulae \eqref{eq:HaanEx} and \eqref{eq:HaanLocal}), a very large class of tensor fields satisfies it.
In Theorem \ref{TH} we prove the following

\textbf{Main Result.}
Given a semisimple Abelian  Haantjes algebra $\mathscr{H}$ on a manifold $M$, that is   an algebra of semisimple commuting  operator fields  with vanishing Haantjes torsion, there exist sets of local coordinates on $M$ in which all the operators can be simultaneously diagonalized.

Conversely, let $\mathcal{K}$  be a family of commuting semisimple operator fields. If they share a set of local coordinates in which they take simultaneously a diagonal form, then they generate a semisimple Abelian Haantjes algebra.

The previous result can also be extended to the very general (but much less explored) case of \textit{non-semisimple} Haantjes operators. Indeed, for this class we shall prove that there exists a local coordinate system where all the operators of a set  of commuting Haantjes operators acquire simultaneously a  block-diagonal form.

In this article a new, infinite ``tower'' of \textit{generalized Nijenhuis torsions} of level $n$ for all $n\in\mathbb{N}$ is also defined. The geometrical meaning of our notion, which naturally generalizes both the classical
Nijenhuis and Haantjes torsions,  has been discussed in detail in Ref.~\cite{TT2019prepr}.

The paper is organized as follows. After a discussion, in Section 2, of the main algebraic structures relevant for this work, including the generalized torsions,  we present in Section 3 a brief introduction to the geometry of Nijenhuis and Haantjes tensors. Also, some new results concerning non-semisimple Haantjes operators are proposed. They allow us to derive the classical Haantjes theorem \cite{Haa} from a new perspective. In Section 4, we discuss the block diagonalization of commuting Haantjes operators. In Section 5, the formal construction of a Haantjes algebra is introduced and the cyclic case is discussed. The results concerning the simultaneous diagonalization of semisimple commuting Haantjes operators are also proved. In Section 6, the Haantjes structure for the classical Coulomb--Kepler system is discussed. A related example, illustrating the application of the main theorems and  the theory of cyclic generators of semisimple Haantjes algebras, is presented in Section 7.

A comparative discussion of other geometric structures related with Haantjes algebras is proposed in the final Section 8.

\section{Nijenhuis  and Haantjes operators}

\label{sec:1}
 The \textit{natural frames} of vector fields associated with local coordinates on a differentiable manifold, being obviously integrable, can be characterized in a \textit{tensorial manner} as  eigen-distributions of a suitable class of  $(1,1)$ tensor fields, i.e.~the ones with vanishing Nijenhuis or Haantjes torsion. In this section, we review some basic algebraic results concerning the theory of such tensors. For a more complete treatment, see
the original papers~\cite{Haa,Nij}  and the related ones~\cite{Nij2,FN}.
\subsection{Algebraic preliminaries}

Let $M$ be a real differentiable manifold and $\boldsymbol{L}:TM\rightarrow TM$ a smooth $(1,1)$ tensor field, i.e.,~a field of linear operators on the tangent space at each point of $M$. In the following, all tensors will be assumed to be smooth.
\begin{definition}
The
 \textit{Nijenhuis torsion} of $\boldsymbol{L}$ is  defined by the vector-valued $2$-form
\begin{equation}
 \label{eq:Ntorsion}
\mathcal{T}_ {\boldsymbol{L}} (X,Y):= \boldsymbol{L}^2[X,Y] +[\boldsymbol{L}X,\boldsymbol{L}Y]-\boldsymbol{L}\Bigl([X,\boldsymbol{L}Y]+[\boldsymbol{L}X,Y]\Bigr),
\end{equation}
where $X,Y \in TM$ and $[ \ , \ ]$ denotes the commutator of two vector fields.
\end{definition}
In local coordinates $\boldsymbol{x}=(x^1,\ldots, x^n)$, the Nijenhuis torsion can be written as the skew-symmetric  $(1,2)$ tensor field
\begin{equation}
\label{eq:NtorsionLocal}
(\mathcal{T}_{\boldsymbol{L}})^i_{jk}=\sum_{\alpha=1}^n\biggl(\frac{\partial {\boldsymbol{L}}^i_k} {\partial x^\alpha} {\boldsymbol{L}}^\alpha_j -\frac{\partial {\boldsymbol{L}}^i_j} {\partial x^\alpha} {\boldsymbol{L}}^\alpha_k+\Bigl(\frac{\partial {\boldsymbol{L}}^\alpha_j} {\partial x^k} -\frac{\partial {\boldsymbol{L}}^\alpha_k} {\partial x^j}\Bigr) {\boldsymbol{L}}^i_\alpha \biggr)\ ,
\end{equation}
having  $n^2(n-1)/2$ independent components.
\begin{definition}
 The \textit{Haantjes torsion} associated with $\boldsymbol{L}$ is the vector-valued $2$-form defined by
\begin{equation}
 \label{eq:Haan}
\mathcal{H}_{\boldsymbol{L}}(X,Y) := \boldsymbol{L}^2\mathcal{T}_{\boldsymbol{L}}(X,Y)+\mathcal{T}_{\boldsymbol{L}}(\boldsymbol{L}X,\boldsymbol{L}Y)-\boldsymbol{L}\Bigl(\mathcal{T}_{\boldsymbol{L}}(X,\boldsymbol{L}Y)+\mathcal{T}_{\boldsymbol{L}}(\boldsymbol{L}X,Y)\Bigr).
\end{equation}
\end{definition}
Explicitly, one can also write~\cite{YKS}
\begin{eqnarray}
 \label{eq:HaanEx}
\nn \mathcal{H}_{\boldsymbol{L}}(X,Y)&=&\boldsymbol{L}^4 [X,Y] -2\boldsymbol{L}^3\Bigl([X,\boldsymbol{L}Y]+[\boldsymbol{L}X , Y]\Bigr)
+\boldsymbol{L}^2\Bigl( [X, \boldsymbol{L}^2 Y]+4\, [\boldsymbol{L}X,\boldsymbol{L}Y]+[\boldsymbol{L}^2X,Y]\Bigr)\\ \nn
&-&2 \boldsymbol{L}\Bigl([\boldsymbol{L}X,\boldsymbol{L}^2Y]+[\boldsymbol{L}^2X,\boldsymbol{L}Y]
\Bigr)+[\boldsymbol{L}^2X,\boldsymbol{L}^2Y] 
\ . \\
\end{eqnarray}

The skew-symmetry of the Nijenhuis torsion implies that the Haantjes torsion is also skew-symmetric.
Its local expression is
\begin{equation}
\label{eq:HaanLocal}
(\mathcal{H}_{\boldsymbol{L}})^i_{jk}=  \sum_{\alpha,\beta=1}^n\biggl(
\boldsymbol{L}^i_\alpha \boldsymbol{L}^\alpha_\beta(\mathcal{T}_{\boldsymbol{L}})^\beta_{jk}  +
(\mathcal{T}_{\boldsymbol{L}})^i_{\alpha \beta}\boldsymbol{L}^\alpha_j \boldsymbol{L}^\beta_k-
\boldsymbol{L}^i_\alpha\Bigl( (\mathcal{T}_{\boldsymbol{L}})^\alpha_{\beta k} \boldsymbol{L}^\beta_j+
 (\mathcal{T}_{\boldsymbol{L}})^\alpha_{j \beta } \boldsymbol{L}^\beta_k \Bigr)
 \biggr) \ ,
\end{equation}
or in explicit form
\begin{eqnarray}
 \label{eq:HaanExCoord}
(\mathcal{H}_{\boldsymbol{L}})^i_{jk}&=&  \sum_{\alpha=1}^n
\biggl(-2  (\boldsymbol{L}^3)^i_\alpha \partial_{\mbox{[}j} \boldsymbol{L}^\alpha_{k\mbox{]}}+ (\boldsymbol{L}^2)^i_\alpha\Bigl( \partial_{\mbox{[}j} (\boldsymbol{L}^2) ^\alpha_{k\mbox{]}}+4  \sum_{\beta=1}^n\boldsymbol{L}^\beta_{\mbox{[}j} \partial_{\lvert \beta\rvert } \boldsymbol{L}^\alpha_{k\mbox{]}}\Bigr) \\
\nn &-&2 \boldsymbol{L}^i_{\alpha} \Bigr(\boldsymbol{L}^\beta_{\mbox{[}j} \partial_{\lvert \beta\rvert } (\boldsymbol{L}^2)^\alpha_{k\mbox{]}} +  (\boldsymbol{L}^2)^\beta_{\mbox{[}j}\partial_{\lvert \beta \rvert } (\boldsymbol{L}) ^{\alpha}_{k\mbox{]}}\Bigr)
+(\boldsymbol{L}^2)^\alpha_{\mbox{[}j}\partial_{\lvert \alpha \rvert } (\boldsymbol{L}^2) ^i_{k\mbox{]}}
\biggr) \ .
\end{eqnarray}
Here for the sake of brevity we have used the notation $\partial_j := \frac{\partial}{\partial x^j}$; the indices between square brackets are to be skew-symmetrized, except those in $\lvert \cdot\rvert $.

In order to prove some of the results of the following sections,  we   state preliminarily three  lemmas.

\begin{lemma}\label{lm:Fond}
Let $\boldsymbol{L}:TM \to TM$ be an   operator field.
 The following identities hold:
\begin{eqnarray}
\label{eq:NijAkerA}
{\tau}_{\boldsymbol{L} }(X ,Y)&=&\boldsymbol{L}^2 [X,Y]\ , \qquad
\qquad\quad \forall X, Y \in \ker(\boldsymbol{L})
\\
\label{eq:HaanAkerA}
\mathcal{H}_{\boldsymbol{L} }(X ,Y)&=&\boldsymbol{L}^4 [X,Y]\ , \qquad
\qquad\quad \forall X, Y \in \ker(\boldsymbol{L})
\end{eqnarray}
\end{lemma}
\begin{proof}
These relations come directly from Eqs. \eqref{eq:Ntorsion} and \eqref{eq:HaanEx}, taking into account that only the first term is non-vanishing in their right hand sides, due to the fact that
$X,Y \in \ker(\boldsymbol{L} )$.
\end{proof}

In the following result, we apply the classical Fitting lemma~\cite{Jacob} to the module, over the ring $\mathcal{C}^{\infty}(M)$, of the operator fields on $TM$.

\begin{lemma}\label{lm:Fitting}
Let   $\boldsymbol{L}:TM \to TM$ be an   operator field. For any $\boldsymbol{x}\in M$ there exists a minimum positive integer $\rho(\boldsymbol{x})$ such that
\begin{equation}
 \label{eq:Riesz}
\ker \boldsymbol{L}^{\rho(\boldsymbol{x}) }(\boldsymbol{x}) = \ker \boldsymbol{L}^{\rho(\boldsymbol{x}) +1} (\boldsymbol{x}) \ ,\quad Im \boldsymbol{L}^{\rho(\boldsymbol{x}) }(\boldsymbol{x}) = Im \boldsymbol{L}^{\rho(\boldsymbol{x}) +1} (\boldsymbol{x})
\end{equation}
and
\begin{equation}
T_{\boldsymbol{x}}M=\ker\boldsymbol{L}^{\rho(\boldsymbol{x}) }(\boldsymbol{x}) \oplus Im\boldsymbol{L}^{\rho(\boldsymbol{x}) }(\boldsymbol{x})
\end{equation}
\end{lemma}
The index $\rho(\boldsymbol{x}) $ is  said to be the Riesz index of $\boldsymbol{L}$ at $\boldsymbol{x}$.

\begin{lemma}\label{Lemma:L5}
Let $\boldsymbol{A}$  and $\boldsymbol{B}$ be two   commuting operator fields  with  Riesz index $\rho(\boldsymbol{x})=1 $ at any point $\boldsymbol{x}$ of $M$. Then, the identity
\begin{equation}
\label{eq:rankAB}
\ker\boldsymbol{A} + \ker\boldsymbol{B} =
\ker\boldsymbol{A} \boldsymbol{B}
\end{equation}
holds.
\end{lemma}
\begin{proof}
The condition $\rho(\boldsymbol{x})=1 $ is equivalent to the conditions
$\ker\boldsymbol{A}=\ker\boldsymbol{A}^2$ and $\ker\boldsymbol{B}=\ker\boldsymbol{B}^2$. Besides, since the two operator fields $\boldsymbol{A}$ and $\boldsymbol{B}$ commute with each others, it is immediate to  check that each vector field $Z:= X+Y$ with $X\in \ker\boldsymbol{A} $ and $Y\in\ker\boldsymbol{B} $ belongs to $ \ker\boldsymbol{AB} $.

Conversely, let us suppose that $\ker\boldsymbol{A}=\ker\boldsymbol{A}^2$. Then, by means of Lemma~\ref{lm:Fitting},
 each $Z\in TM$ can be decomposed  as $Z:= X+ \boldsymbol{A}W$ with $X\in \ker \boldsymbol{A}$ and some $W\in TM$. Let us prove that for each $Z\in \ker\boldsymbol{A} \boldsymbol{B}$, the vector field $\boldsymbol{A}W \in \ker \boldsymbol{B}$. In fact,  taking into account that $\boldsymbol{A}$ and $\boldsymbol{B}$ commute, the condition $ \boldsymbol{A} \boldsymbol{B}Z= \mathbf{0}$ implies that $  \boldsymbol{B} \boldsymbol{A}^2 W=\mathbf{0}$, i.e.,~$\boldsymbol{B}  W\in \ker \boldsymbol{A}^2=\ker \boldsymbol{A}$. Therefore, $\boldsymbol{A}W \in \ker \boldsymbol{B}$.
\end{proof}
\begin{corollary}
Under the same hypotheses of Lemma~\ref{Lemma:L5}, we have
\begin{equation}
 \label{eq:cor6}
\ker\boldsymbol{A}\boldsymbol{B} =
\ker\boldsymbol{A}^2 \boldsymbol{B}^2 \ .
\end{equation}
\end{corollary}

\subsection{A new family of generalized Nijenhuis torsions}

We shall now generalize the notion of Haantjes torsion by means of a recursive procedure. Indeed, one can introduce a ``\textit{tower}'' of generalized torsions of Nijenhuis type.
\begin{definition}
We define the univariate generalized Nijenhuis torsion of level $n$ as the vector-valued 2-form given by
\begin{equation}
{\tau}^{(n)}_{\boldsymbol{L}}(X,Y):= \boldsymbol{L}^2{\tau}^{(n-1)}_{\boldsymbol{L}}(X,Y)+{\tau}^{(n-1)}_{\boldsymbol{L}}(\boldsymbol{L}X,\boldsymbol{L}Y)-
\boldsymbol{L}\Bigl({\tau}^{(n-1)}_{\boldsymbol{L}}(X,\boldsymbol{L}Y)+{\tau}^{(n-1)}_{\boldsymbol{L}}(\boldsymbol{L}X,Y)\Bigr), n\geq 1
\end{equation}
where $\tau_{\boldsymbol{L}}^{(0)}(X,Y)= [X,Y]$, $X,Y \in TM$. Here $\tau_{\boldsymbol{L}}^{(1)}=\tau_{\boldsymbol{L}}$ and $\tau_{\boldsymbol{L}}^{(2)}=\mathcal{H}_{\boldsymbol{L}}$.
\end{definition}
The expression of the $n$th-level torsion in local coordinates is given by
\begin{equation}
\label{eq:GenHaanLocal}
({\tau}^{(n)}_{\boldsymbol{L}})^i_{jk}=  \sum_{\alpha,\beta=1}^n\biggl(
\boldsymbol{L}^i_\alpha \boldsymbol{L}^\alpha_\beta({\tau}^{(n-1)}_{\boldsymbol{L}})^\beta_{jk}  +
({\tau}^{(n-1)}_{\boldsymbol{L}})^i_{\alpha \beta}\boldsymbol{L}^\alpha_j \boldsymbol{L}^\beta_k-
\boldsymbol{L}^i_\alpha\Bigl( ({\tau}^{(n-1)}_{\boldsymbol{L}})^\alpha_{\beta k} \boldsymbol{L}^\beta_j+
 ({\tau}^{(n-1)}_{\boldsymbol{L}})^\alpha_{j \beta } \boldsymbol{L}^\beta_k \Bigr)
 \biggr) \ .
\end{equation}
We wish to recall that a notion of generalized Nijenhuis torsions was also proposed, in a independent way, in~\cite{YKS}.
In that algebraic construction, one considers generalized torsions of order $n$ associated more generally to an arbitrary vector-valued skew symmetric bilinear map on a real vector space.
In our geometric approach, which is a recursive one, we work on a tangent bundle and fix the initial condition of the recurrence
with the standard choice of the Lie bracket between vector fields.

\subsection{General properties of Haantjes operators}

We shall first consider some specific cases, in which the construction of the Nijenhuis and Haantjes torsions is particularly simple.

\begin{example}\label{ex:Nd}
Let $\boldsymbol{L}: TM \to TM$  be an operator field that takes the diagonal form
\begin{equation}
\boldsymbol{L}(\boldsymbol{x})=\sum _{i=1}^n l_{i }(\boldsymbol{x}) \frac{\partial}{\partial x^i}\otimes \rd x^i \label{eq:Ldiagonal}
\end{equation}
in some local chart $\boldsymbol{x}=(x^1,\ldots,x^n)$.
The components of its Nijenhuis torsion read
\begin{equation}
 \label{eq:Nd}
(\mathcal{T}_{\boldsymbol{L}})^i_{jk}=(l_j-l_k)
\left( \frac{\partial l_j}{\partial x^k}\delta^i_j+\frac{\partial l_k}{\partial x^j}\delta^i_k\right).
\end{equation}
It is evident  that  $(\mathcal{T}_{\boldsymbol{L}})^i_{jk}=0$ if
 $i$, $j$ and $k$ are all distinct or if $j=k$. Thus, we can limit ourselves to analyze the $n(n-1)$ components
\begin{equation}
(\mathcal{T}_{\boldsymbol{L}})^j_{jk}=(l_j-l_k)\frac{\partial l_j}{\partial x^k} \ ,
\qquad\qquad j\neq k \ .
\end{equation}
 If $\frac{\partial l_j}{\partial x^k}\neq 0$, each component vanishes if and only if $l_j(\boldsymbol{x})\equiv l_k(\boldsymbol{x})$.
Therefore, we can state the following
\end{example}
\begin{lemma}\label{lem:LN}
Let $\boldsymbol{L}$ be the diagonal  operator field \eqref{eq:Ldiagonal}, and assume that  its  Nijenhuis torsion  vanishes. Let us denote by
$(i_1,\ldots,i_j,\ldots, i_r)$, $r\leq n$ an ordered subset of $(1,2,\ldots,n)$. If the $i_j$th eigenvalue of  $\boldsymbol{L}$ depends on the variables  $(i_1,\ldots,i_j,\ldots, i_r)$, then
\begin{equation}
l_{i_j}(x^{i_1},\ldots,x^{i_j},\ldots, x^{i_r}) = l_{i_1}= l_{i_2}=\ldots  = l_{i_r} \label{eq:eigenvalues} \ .
\end{equation}
\end{lemma}

In addition to the trivial case when each eigenvalue is constant, we can distinguish several further cases, in which the Nijenhuis torsion of a diagonal operator vanishes. For instance,
\begin{itemize}
\item[(i)]
  $l_j(\boldsymbol{x})=\lambda_j (x^j)\qquad\qquad\qquad  j=1,\ldots,n \Longrightarrow \textit{n simple eigenvalues}$
\item[(ii)]
$l_j(\boldsymbol{x})=\lambda(\boldsymbol{x}) \qquad\qquad\qquad\quad  j=1,\ldots,n \Longrightarrow  \textit{1 eigenvalue of  multiplicity n}$ \ .

\end{itemize}
An exhaustive analysis of all the remaining, intermediate possibilities is left to the reader.

\begin{example}\label{ex:H2}
Let $dim~M=2$. It is easy to prove by a straightforward computation that the Haantjes torsion of any field of smooth operators vanishes.

\end{example}

\begin{example}\label{ex:Hd}
Let $\boldsymbol{L}$ be the diagonal operator of Eq.~\eqref{eq:Ldiagonal}. The components of its Haantjes torsion read
\begin{equation}
(\mathcal{H}_{\boldsymbol{L}})^i_{jk}=(l_i-l_j )(l_i-l_k)(\mathcal{T}_{\boldsymbol{L}})^i_{jk} = 0, \label{Haantdiag}
\end{equation}
where $(\mathcal{T}_{\boldsymbol{L}})^i_{jk} $ is given by Eq.~\eqref{eq:Nd}.
\end{example}

\begin{proposition}\label{pr:Hd}
Let $\boldsymbol{L}: TM\to TM$ be an operator field. If there exists a local coordinate chart $(x^1,\ldots, x^n)$ where $\boldsymbol{L}$ takes the diagonal form \eqref{eq:Ldiagonal}, then
the Haantjes torsion of $\boldsymbol{L}$ identically vanishes.
\end{proposition}

Due to the relevance of the  Haantjes (Nijenhuis) vanishing condition,
we propose the following

\begin{definition}
A Haantjes (Nijenhuis)   operator is an operator field whose  Haantjes (Nijenhuis) torsion identically vanishes.
\end{definition}

We also recall that the transposed operator  $\boldsymbol{L}^{T}: T^*M \mapsto T^*M$ is  defined to be the transposed linear map of $\boldsymbol{L}$ with respect to the natural pairing between the tangent and the cotangent bundle of $M$:
\[
\langle \boldsymbol{L}^T \alpha , X \rangle =\langle \alpha ,\boldsymbol{L}X \rangle  \qquad\qquad\alpha \in T^*M, \ X\in TM.
\]
 The condition for a Nijenhuis operator $\boldsymbol{N}$ to be  torsionless can be expressed in terms of its Lie derivative along the flow of any vector field $X\in TM$  in the following, equivalent  way  (see, for instance,~\cite{GVY}):
\begin{equation}
\label{eq:TNLie}
  \mathcal{L}_ {\boldsymbol{N}X}( \boldsymbol{N})=\boldsymbol{N}\mathcal{L}_ {X}( \boldsymbol{N}) \ .
\end{equation}

Analogously, the vanishing of the Haantjes torsion \eqref{eq:Haan} of an operator field $\boldsymbol{L}$ is equivalent to the novel condition
\begin{equation}
 \mathcal{L}_ {\boldsymbol{L}^2X}\bigl( \boldsymbol{L}\bigr)\boldsymbol{L}=\boldsymbol{L}^3\mathcal{L}_ {X}\bigl( \boldsymbol{L}\bigr) -
\boldsymbol{L}^2
\Bigl(2 \mathcal{L}_ {\boldsymbol{L}X} (\boldsymbol{L})+\mathcal{L}_ {X} (\boldsymbol{L})\boldsymbol{L}\Bigr)+
\boldsymbol{L}\Bigl(  \mathcal{L}_ {\boldsymbol{L}^2X}( \boldsymbol{L})+2\mathcal{L}_ {\boldsymbol{L}X} (\boldsymbol{L})\boldsymbol{L}\Bigr).
\end{equation}

As is well known (see for instance~\cite{GVY}),  given  an invertible Nijenhuis operator, its inverse is also a Nijenhuis operator. The same property holds true for a Haantjes operator.
\begin{proposition}\label{pr:Linv}
Let $\boldsymbol{L}: TM \to TM$ be a Haantjes operator. If $\boldsymbol{L}^{-1}$ exists, it is also a Haantjes operator.
\end{proposition}
\begin{proof}
It is a consequence of the identity
\begin{equation}
\mathcal{H}_{\boldsymbol{L}^{-1}}(X,Y)=
\boldsymbol{L}^{-4}\ \mathcal{H}_{\boldsymbol{L}}(\boldsymbol{L}^{-2}X,\boldsymbol{L}^{-2}Y)
\end{equation}
that can be easily deduced from Eq.~\eqref{eq:HaanEx}. For an alternative proof, see Proposition 2, p. 257 of~\cite{BogCMP}.
\end{proof}

The product of a Nijenhuis operator with a generic function is no longer a Nijenhuis operator, as
is proved by  the identity
\begin{equation}
 \label{eq:fNtorsion}
\mathcal{T}_{f \boldsymbol{L}}(X,Y)=f^2\mathcal{T}_{\boldsymbol{L}}(X,Y)+f\Bigl((\boldsymbol{L}X)(f) \boldsymbol{L}Y-(\boldsymbol{L}Y)(f) \boldsymbol{L}X+Y(f)\boldsymbol{L}^2 X-X(f)\boldsymbol{L}^2 Y\Bigr) \ ,
\end{equation}
where $X(f)$ denotes the Lie derivative of an arbitrary function $f \in C^\infty(M)$  with respect to the vector field $X$.
However, the differential and algebraic properties of the Haantjes operators are much richer, as one can infer from the following, remarkable results.

\begin{proposition}[\cite{BogCMP}]\label{pr:fL}
Let  $\boldsymbol{L}: TM\to TM$ be an  operator field, $f,g:M \rightarrow \mathbb{R}$ two $C^\infty(M)$ functions, and $\boldsymbol{I}$  the identity operator in $TM$. Then we have
\begin{equation}
 \label{eq:LtorsionLocal}
\mathcal{H}_{f \boldsymbol{I}+g \boldsymbol{L}}(X,Y)=g^4\, \mathcal{H}_{ \boldsymbol{L}}(X,Y) \ .
\end{equation}
\end{proposition}
\begin{proof}
See Proposition 1, p. 255 of~\cite{BogCMP}.
\end{proof}
\begin{proposition}[\cite{BogI}]\label{pr:Lpowers}
Let $\boldsymbol{L}: TM \to TM$ be a Haantjes operator. Then, for any  polynomial in $\boldsymbol{L}$ with coefficients $a_{j}\in C^\infty(M)$, the associated Haantjes torsion vanishes, i.e.~\begin{equation}
\mathcal{H}_{\boldsymbol{L}}(X,Y)= \mathbf{0} \ \Longrightarrow \
\mathcal{H}_{(\sum_j a_{j} (\boldsymbol{x}) \boldsymbol{L}^j)}(X,Y)= \mathbf{0}.
\end{equation}
\end{proposition}
\begin{proof}
See Corollary 3.3, p. 1136 of~\cite{BogI}.
\end{proof}
As we shall see in Section 5.2, as a consequence of  Propositions \ref{pr:fL} and \ref{pr:Lpowers}, it follows that a Haantjes operator generates a cyclic Haantjes algebra (i. e. a cyclic algebra of Haantjes operators) over the ring of smooth functions on $M$.    However, this is not the case for a Nijenhuis operator  $\boldsymbol{N}$, since a polynomial in $\boldsymbol{N}$ with coefficients $a_j \in C^\infty(M)$ is not necessarily a Nijenhuis operator.

\subsection{An application to Classical Mechanics: the inertia tensor}
 We wish to present a new  example of application of Nijenhuis and Haantjes operators, borrowed from Classical Mechanics.

\begin{example}\label{ex:Hin}
Let $ \mathcal{M}=\{(P_\gamma, m_\gamma)\in (\mathcal{E}_n,\mathbb{R})\} $ be a finite system of mass points (possibly with $m_\gamma <0$) in
 the $n$-dimensional affine Euclidean space $\mathcal{E}_n$. Let us consider the $(1,1)$  tensor field defined by
\begin{equation}
\boldsymbol{E}_P(\vec{v})=\sum _\gamma m_\gamma\bigl( (P_\gamma-P) \cdot\vec{v}\bigr) \ (P_\gamma-P)
\qquad \vec{v}\in T_P\mathcal{E}_n\equiv \mathbb{E}_n\ ,
\end{equation}
known as the \textit{planar inertia tensor}  (or Euler tensor in Continuum Mechanics). The  \textit{inertia tensor} is given by
\begin{equation}
 \label{eq:InT}
{\mathbb{I}}_P(\vec{v})=\sum _\gamma m_\gamma\biggl( \lvert P_\gamma-P\rvert ^2\vec{v}-\bigl((P_\gamma-P) \cdot\vec{v}\bigr) \ (P_\gamma-P)\biggr) \ .
\end{equation}
  They are related by the formulae
\begin{equation}
 \label{eq:relEI}
\mathbb{I}_P=Trace( \boldsymbol{E}_P) \boldsymbol{I}_n  -\boldsymbol{E}_P\ ,\qquad
\boldsymbol{E}_P= \frac{Trace(\mathbb{I}_P)}{n-1}\boldsymbol{I}_n  -\mathbb{I}_P \ , \qquad n>1\ ,
\end{equation}
where $\boldsymbol{I}_n$ is the identity operator in
$\mathbb{E}_n$.  Both of them are symmetric w.r.t. the Euclidean scalar product, so that they are diagonalizable at any point of $\mathcal{E}_n$. Furthermore, by virtue of relations \eqref{eq:relEI} they commute; consequently, they can be pointwise  diagonalized, simultaneously.

If $G$ is the center of mass of  $ \mathcal{M}$, defined by
\[
G-P=\frac{1}{m}\sum_\gamma (P_\gamma-P) \qquad   m := \sum_\gamma m_\gamma \qquad \qquad m\in \mathbb{R}\setminus \{0\},
\]
the following Huygens--Steiner transposition formulae hold
\begin{eqnarray}
\label{eq:Etrans}
\boldsymbol{E}_P(\vec{v})&=&\boldsymbol{E}_G(\vec{v})+ m \bigl((P-G) \cdot\vec{v}\bigr) \ (P-G), \\
\label{eq:Itrans}
\mathbb{I}_P(\vec{v})&=&\mathbb{I}_G(\vec{v})+m \lvert P-G\rvert ^2 -m \bigl((P-G)\cdot\vec{v}\bigr) \ (P-G) \ .
\end{eqnarray}
From Eqs. \eqref{eq:Etrans} and \eqref{eq:Itrans} it follows that
 in the Cartesian coordinates $(x^1,\ldots, x^n)$ with origin  in $G$, defined by the common eigen-directions of $E_G$ and $\mathbb{I}_G$, we have
\begin{eqnarray}
 (\boldsymbol{E}_P)^{i}_j&=&\lambda_i(G)\delta^i_{j}+m\, x^{i} x_{j}, \\
 (\mathbb{I}_P)^{i}_j&=&l_i(G)\delta^i_{ j}+m\biggl(\sum_{\alpha=1}^n x^{\alpha} x_{\alpha}-\, x^{i} x_{j} \biggr) \qquad i,j=1,\ldots,n \ .
\end{eqnarray}
Here $x_{\alpha}= \delta_{\alpha \beta} x^{\beta}$, and $\lambda_i(G)$ and $l_j(G)$ denote the eigenvalues of the
tensor fields  $\boldsymbol{E}$ and $\mathbb{I}$ respectively, both evaluated at the point $G$.
In~\cite{Ben92,Ben93} it has been proved  that the Nijenhuis torsion of $\boldsymbol{E}$ vanishes
    \begin{equation}
\label{eq:Etors} 
(\mathcal{T}_{\boldsymbol{E}})^i_{jk}=
m\sum_{\alpha =1}^n\bigg(x^i\cancel{(\delta_{\alpha k} \boldsymbol{E}^\alpha_j-\delta_{\alpha j}\boldsymbol{E}^\alpha_k)}+
\bcancel{(\delta_{j k}-\delta_{kj})}x^\alpha\boldsymbol{E}^i_\alpha \bigg)=0 \ ,
\end{equation}
thus its Haantjes torsion also  vanishes. Furthermore, we observe that
\begin{equation}
\label{eq:Itors}
(\mathcal{T}_{\mathbb{I}})^i_{jk}=
2m\sum_{\alpha =1}^n\biggl(x_\alpha \bigl(\mathbb{I}^\alpha_j\delta^i_{k}-\mathbb{I}^\alpha_k\delta^i_{j}\bigr)+x_k \mathbb{I}^i_j-x_j\mathbb{I}^i_k\biggr) \ ,
\end{equation}
i.e.~ the Nijenhuis torsion of $\mathbb{I}$ is not identically zero,  although  its  Haantjes torsion vanishes  as a consequence of the identity \eqref{eq:LtorsionLocal}, applied to the first equation  \eqref{eq:relEI}. Thus, the planar inertia tensor $\boldsymbol{E}_P(\vec{v})$ is a Nijenhuis operator, whereas the inertia tensor ${\mathbb{I}}_P(\vec{v})$ is a Haantjes operator.
\end{example}
Other relevant examples of Haantjes operators expressed in terms of Killing tensors in a Riemannian manifold can be found in~\cite{TT2016}.

\section{The geometry of Haantjes operators}

\label{sec:Sp}

\subsection{Integrability}

 As  stated in  Proposition \ref{pr:Hd}, the Haantjes torsion $\mathcal{H}_{\boldsymbol{L}}$ of a semisimple operator field $\boldsymbol{L}$ has a relevant geometrical meaning:
its vanishing is a necessary  condition for the \textit{eigen-distributions} of $\boldsymbol{L}$ to be integrable. To clarify  this point,  first we need  to recall that
a reference frame is a set of $n$ vector fields $\{Y_1,\ldots,Y_n\}$ such that, at each point $\boldsymbol{x}\in U\subseteq M$, $U$ open set, they
form a basis of the tangent space $T_{\boldsymbol{x}}U$.  Two frames  $\{X_1,\ldots,X_n\}$ and $\{Y_1,\ldots,Y_n\}$ are said to be equivalent if $n$ nowhere vanishing smooth
functions $f_i$ exist, such that
\[
 X_i= f_i(\boldsymbol{x}) Y_i \ , \qquad\qquad i=1,\ldots,n \ .
\]
 A \textit{natural} frame is the frame associated to a local chart $\{U, (x^1,\ldots,x^n)\}$ which is formed by the vector fields $\left\{\frac{\partial}{\partial x^1}, \ldots, \frac{\partial}{\partial x^n}\right\}$.
\begin{definition}\label{def:Iframe}
  An \emph{integrable} frame  is a reference frame equivalent to a natural frame.
\end{definition}
\begin{proposition}[\cite{BCR02}]\label{pr:BCRframe}
A reference frame $\{Y_1,\ldots,Y_n \}$ on a manifold $M$ is integrable if and only  one the following two, equivalent conditions are satisfied:
\begin{itemize}
\item
each distribution generated by any two vector fields $\{Y_i,Y_j\}$ is Frobenius integrable;
\item
each distribution $\mathcal{E}_i$ generated by all the vector fields except $Y_i$ is Frobenius integrable.
\end{itemize}
\end{proposition}

\begin{definition}\label{def:Kdiag}
An operator field $\boldsymbol{L}$ is \textit{semisimple}
if, in each open neighborhood $U\subseteq M$, there exists a reference frame formed by proper eigenvector fields of $\boldsymbol{L}$. This frame will be said to be an eigen-frame of $\boldsymbol{L}$. Moreover, $\boldsymbol{L}$ is \textit{simple} if all of its eigenvalues are pointwise distinct, namely $l_i (\boldsymbol{x}) \neq l_j(\boldsymbol{x})$, $i, j=1,\ldots,n$, $ \forall \boldsymbol{x}\in M$.
\end{definition}
An important problem is to establish the conditions ensuring that the eigen-frames of $\boldsymbol{L}$ are integrable. Proposition \ref{pr:Hd} amounts to say that if an  operator admits a local chart in which it takes a diagonal form, then its Haantjes torsion \emph{necessarily} vanishes. Therefore,   the associated natural frame  is an eigen-frame (trivially) integrable.
 In 1955, Haantjes proved in~\cite{Haa} that the vanishing of the  Haantjes torsion of a \emph{semisimple} (with possibly coinciding eigenvalues) operator $\boldsymbol{L}$ is also a \emph{sufficient} condition
to ensure  the integrability of
each  of its  eigen-distributions and their direct sums; consequently,  it also guarantees the existence of local coordinate charts in  which $\boldsymbol{L}$ takes a  diagonal form.
We shall say that such coordinates are a set of \emph{Haantjes coordinates} for $\boldsymbol{L}$. We remark that, in the framework of hydrodynamic systems, when the eigenvalues are simple, the Haantjes coordinates coincide with the \textit{Riemann invariants} of the system. Furthermore, Haantjes stated that the vanishing
of the  Haantjes torsion of an   operator with real eigenvalues $\boldsymbol{L}$ is also a sufficient (but   not necessary) condition to ensure the integrability of  each  of its generalized
eigen-distributions (of constant rank) and their direct sums. An equivalent statement of the above-mentioned results  is that each Haantjes operator with real eigenvalues  admits an  integrable frame of generalized eigenvector fields.

Below, we prove necessary and sufficient conditions for the integrability of generalized eigen-distributions of  \textit{non-semisimple} operator fields.
 Let us denote by $Spec(\boldsymbol{L}):= \{ l_1(\boldsymbol{x}),
 l_2(\boldsymbol{x}), \ldots, l_s(\boldsymbol{x})\}$  the set of the eigenvalues of an operator $\boldsymbol{L}$, which  we shall always assume to be \emph{real} and pointwise distinct in all the forthcoming considerations. Also, we denote by
\begin{equation}
 \label{eq:DisL}
 \mathcal{D}_i = \ker \Bigl(\boldsymbol{L}-l_i\boldsymbol{I}\Bigr)^{\rho_i}, \qquad i=1,\ldots,s
\end{equation}
the $i$th generalized eigen-distribution, that is  the distribution of all the  generalized eigenvector fields  corresponding to the eigenvalue $l_i=l_i(\boldsymbol{x})$. In Eq.~\eqref{eq:DisL}, $\rho_i$ stands for the Riesz index of $l_i$ (or equivalently, the index of the associated eigen-distribution), defined as the minimum integer such that
\begin{equation}
 \label{eq:Riesz}
\ker \Bigl(\boldsymbol{L}-l_i\boldsymbol{I}\Bigr)^{\rho_i} = \ker \Bigl(\boldsymbol{L}-l_i\boldsymbol{I}\Bigr)^{\rho_{i}+1} \ ,
\end{equation}
which  we shall always assume to be independent of $\boldsymbol{x}$ (the Riesz index of $\boldsymbol{L}$, introduced in Lemma 4, corresponds to the case $l_i=0$ ).
When $\rho_i=1$,  $\mathcal{D}_i$ is a proper eigen-distribution.
\begin{definition}\label{def:ss}
  A generalized  \emph{eigen-frame} of an operator field  $\boldsymbol{L}$ is a local reference frame of generalized eigenvector fields of $\boldsymbol{L}$.
\end{definition}
Let us recall that the tangent space at any point of $\boldsymbol{x}\in M$ admits the spectral decomposition
\begin{equation}
 \label{eq:TMscomp}
T_{\boldsymbol{x}}M=\bigoplus_{i=1} ^s \mathcal{D}_i(\boldsymbol{x})\ . \quad \qquad
\end{equation}


\begin{theorem}\label{th:HTT}
Let $\boldsymbol{L}: TM \to TM$ be an   operator field and $\mathcal{D}_i$ one of its generalized eigen-distributions, with Riesz index $\rho_i$, $i\in \{1,\ldots, s\}$. We assume that the rank of  $\mathcal{D}_i$ is independent of $\boldsymbol{x}\in M$. Then, the following three conditions are equivalent:
\begin{enumerate}
\item
The distribution $\mathcal{D}_i$ is involutive;

\item
$
{\tau}_{(\boldsymbol{L}-l_i \boldsymbol{I})^{\rho_i}}(\mathcal{D}_i , \mathcal{D}_i )= \mathbf{0}  \ ;
$

\item
$
\mathcal{H}_{(\boldsymbol{L}-l_i \boldsymbol{I})^{\rho_i}}(\mathcal{D}_i , \mathcal{D}_i )= \mathbf{0}\ .
$

\end{enumerate}
\end{theorem}
\begin{proof}
First, let us prove that (1) $\Leftrightarrow$ (2). From Eqs.~\eqref{eq:DisL} and \eqref{eq:NijAkerA} we obtain
\[
{\tau}_{(\boldsymbol{L}-l_i \boldsymbol{I})^{\rho_i}}(\mathcal{D}_i, \mathcal{D}_i)=
(\boldsymbol{L}-l_i \boldsymbol{I})^{2\rho_i} \bigl[\mathcal{D}_i,\mathcal{D}_i\bigr ] \ .
\]
Therefore, we deduce that
\[
[\mathcal{D}_i,\mathcal{D}_i]
\subseteq \ker \Bigl(\boldsymbol{L}-l_i\boldsymbol{I}\Bigr)^{2\rho_i}\stackrel{\text{\eqref{eq:Riesz}}}=
\ker \Bigl(\boldsymbol{L}-l_i\boldsymbol{I}\Bigr)^{\rho_i}
\]
if and only if condition (2) is fulfilled.
Obviously, condition (2) implies (3). Consequently, condition  (1) implies (3).
Finally, let us prove that condition (3) implies (1).
In fact, from Eqs.~\eqref{eq:DisL} and \eqref{eq:HaanAkerA} it follows that
\[
\mathcal{H}_{(\boldsymbol{L}-l_i \boldsymbol{I})^{\rho_i}}(\mathcal{D}_i, \mathcal{D}_i)=
(\boldsymbol{L}-l_i \boldsymbol{I})^{4\rho_i} \bigl[\mathcal{D}_i,\mathcal{D}_i\bigr] \ .
\]
Therefore,
\[
[\mathcal{D}_i,\mathcal{D}_i]
\subseteq \ker \Bigl(\boldsymbol{L}-l_i\boldsymbol{I}\Bigr)^{4\rho_i}\stackrel{\text{\eqref{eq:Riesz}}}=
\ker \Bigl(\boldsymbol{L}-l_i\boldsymbol{I}\Bigr)^{\rho_i} \ ,
\]
if  condition (3) is fulfilled.
\end{proof}
\begin{theorem}\label{pr:FXmn}
Let $\boldsymbol{L}: TM \to TM$ be an operator field, and $\mathcal{D}_i$, $  \mathcal{D}_j$  two   involutive eigen-distributions with Riesz indices $\rho_i$ and $\rho_j$,
respectively ($i,j \in \{1,\ldots, s\}$). Then, the following three conditions are equivalent:
\begin{enumerate}
\item
the distribution $\mathcal{D}_i  \oplus \mathcal{D}_j$ is involutive;
\item
$
{\tau}_{(\boldsymbol{L}-l_i \boldsymbol{I})^{\rho_i}(\boldsymbol{L}-l_j \boldsymbol{I})^{\rho_j}}(\mathcal{D}_i  \oplus \mathcal{D}_j, \mathcal{D}_i  \oplus \mathcal{D}_j)=\mathbf{0};
$

\item
$
\mathcal{H}_{(\boldsymbol{L}-l_i \boldsymbol{I})^{\rho_i}(\boldsymbol{L}-l_j \boldsymbol{I})^{\rho_j}}(\mathcal{D}_i  \oplus \mathcal{D}_j, \mathcal{D}_i  \oplus \mathcal{D}_j)=\mathbf{0} \ .
$
\end{enumerate}
\end{theorem}
\begin{proof}
Firstly, let us prove that (1) $\Leftrightarrow$ (2).
From Eq.~\eqref{eq:NijAkerA} applied to $\boldsymbol{A} := (\boldsymbol{L}-l_i \boldsymbol{I})^{\rho_i}$, $\boldsymbol{B} :=  (\boldsymbol{L}-l_j \boldsymbol{I})^{\rho_j}$ and  \eqref{eq:rankAB} it follows that
\[
{\tau}_{\boldsymbol{A} \boldsymbol{B}}(\mathcal{D}_i  \oplus \mathcal{D}_j, \mathcal{D}_i  \oplus \mathcal{D}_j)= \boldsymbol{A}^2\boldsymbol{B}^2[\mathcal{D}_i  \oplus \mathcal{D}_j, \mathcal{D}_i  \oplus \mathcal{D}_j]=
\boldsymbol{A}^2 \boldsymbol{B}^2
[\mathcal{D}_i,\mathcal{D}_j]\ .
\]
Therefore,
\[
[\mathcal{D}_i,\mathcal{D}_j]
\subseteq \ker \boldsymbol{A}^{2}\boldsymbol{B} ^2\stackrel{\text{\eqref{eq:Riesz},\eqref{eq:cor6}}}
=\ker  \boldsymbol{A} \boldsymbol{B} \stackrel{\text{\eqref{eq:rankAB}}}=\ker \boldsymbol{A}\oplus \ker \boldsymbol{B}\ \ ,
\]
if and only if  condition (2) is fulfilled. Obviously, condition (2) implies (3), therefore (1) implies (3).
Finally,  from Eq.~\eqref{eq:rankAB} and \eqref{eq:HaanAkerA} it follows that
\[
\mathcal{H}_{\boldsymbol{A}\boldsymbol{B}}(\mathcal{D}_i  \oplus \mathcal{D}_j, \mathcal{D}_i  \oplus \mathcal{D}_j)= \boldsymbol{A}^4\boldsymbol{B}^4[\mathcal{D}_i  \oplus \mathcal{D}_j, \mathcal{D}_i  \oplus \mathcal{D}_j]=
\boldsymbol{A}^4 \boldsymbol{B}^{4}
[\mathcal{D}_i,\mathcal{D}_j]\ .
\]
Consequently
\[
[\mathcal{D}_i,\mathcal{D}_j]
\subseteq \ker \boldsymbol{A}^4\boldsymbol{B}^4\stackrel{\text{\eqref{eq:Riesz},\eqref{eq:cor6}}}=
\ker  \boldsymbol{A} \boldsymbol{B} \stackrel{\text{\eqref{eq:rankAB}}}=\ker \boldsymbol{A}\oplus \ker \boldsymbol{B} \ , 
\]
if  condition  (3) is satisfied.
\end{proof}

\subsection{A new proof of the Haantjes theorem}

\begin{definition}\label{def:mI}
Let us consider a set of distributions $\{\mathcal{D}_i, \mathcal{D}_j, \ldots, \mathcal{D}_k \}$.
We shall say that such distributions  are \textit{mutually integrable} if

(i) each of them is integrable;

(ii) any  sum $\mathcal{D}_i  + \mathcal{D}_j +\cdots +\mathcal{D}_k$ (where all indices $i,j,\ldots, k$ are different) is also integrable.
\end{definition}
Now, we are in the position to recover the Haantjes theorem as a consequence of   Theorems \ref{th:HTT} and \ref{pr:FXmn}.

\begin{theorem}[\cite{Haa}]\label{th:Haan}
Let $\boldsymbol{L}: TM \to TM$ be an   operator field, and assume that the rank of each generalized eigen-distribution $\mathcal{D}_i$, $i=1,\ldots, s$ is independent of $\boldsymbol{x}\in M$.
The vanishing of the Haantjes torsion
\begin{equation}
 \label{eq:HaaNullTM}
\mathcal{H}_{\boldsymbol{L}}(X, Y)= \mathbf{0} \qquad\qquad\qquad \forall ~ X, Y \in TM
\end{equation}
is a  sufficient condition to ensure the mutual integrability of the  generalized eigen-distributions $\{\mathcal{D}_1,\ldots, \mathcal{D}_s\}$.
In addition, if $\boldsymbol{L}$ is  semisimple, condition \eqref{eq:HaaNullTM} is also necessary.
\end{theorem}
\begin{proof}
 In the non-semisimple case, according to Proposition \ref{pr:Lpowers}, if  condition  \eqref{eq:HaaNullTM} is satisfied, then the  conditions (3) of Theorems  \ref{th:HTT} and \ref{pr:FXmn}  hold for $i=1,\ldots,s$; therefore, the distributions $\{\mathcal{D}_1,\ldots, \mathcal{D}_s\}$ are mutually integrable.

Conversely, if the   distributions  $\{\mathcal{D}_1,\ldots, \mathcal{D}_s\}$
are mutually integrable, then the conditions (3) of Theorems \ref{th:HTT} and \ref{pr:FXmn}  are satisfied for all eigen-distributions $\mathcal{D}_i$.  In addition, if $\boldsymbol{L}$ is semisimple, then each Riesz index $\rho_i=1$, and Proposition \ref{pr:fL}  implies that
\begin{equation}
 \label{eq:Haux}
\mathcal{H}_{\boldsymbol{L}-l_i \boldsymbol{I}}(X , Y)= \mathcal{H}_{\boldsymbol{L}}(X , Y) ,  \qquad\qquad
\forall ~ X, Y \in TM , \qquad i=1,\ldots,s \ .
\end{equation}
Therefore, due to the spectral decomposition \eqref{eq:TMscomp}, we deduce that condition \eqref{eq:HaaNullTM} is fulfilled.
\end{proof}
\begin{remark}
Due to Theorems \ref{th:HTT} and \ref{pr:FXmn}, conditions (2) are  not stronger than conditions (3) but equivalent. This can be understood by observing that if a distribution $\mathcal{D}_i $ is integrable, then the operator $(\boldsymbol{L}-l_i \boldsymbol{I})^{\rho_i}$ can be restricted to each integral leaf of $\mathcal{D}_i $ and such restriction vanishes. Thus, one might wonder, at least in the semisimple case,  whether there exists an analogous of Proposition \ref{pr:fL} for the  Nijenhuis torsion. This is not the case as, unlike Eq.~\eqref{eq:Haux}, we have
\begin{equation}
\label{eq:Taux}
{\tau}_{\boldsymbol{L}-l_i \boldsymbol{I}}(X , Y)={\tau}_{\boldsymbol{L}}(X , Y)
+ ( \boldsymbol{L}Y)(l_i)X-( \boldsymbol{L}X)(l_i)Y
\ ,
\end{equation}
which does not coincide with ${\tau}_{\boldsymbol{L}}(X , Y)$ unless the eigenvalue $l_i$ is constant.
\end{remark}

In the article by Haantjes~\cite{Haa}, the original proof of  Theorem \ref{th:Haan}, which was developed in a completely different way, is explicitly carried out  only for the case of a semisimple operator.
In~\cite{DeF},  the integrability of the eigen-distributions of a \textit{Nijenhuis operator} with generalized eigenvectors of Riesz index $2$ was proved. However, the case of Haantjes operators was not
considered. Besides, to the best of our knowledge, the proofs of the Haantjes theorem available in the literature (see for instance~\cite{FN,GVY}) are based on the quite restrictive assumption that the Haantjes operator is semisimple.

Let us show in  detail how we can  determine a coordinate system that, under  the assumptions of Theorem \ref{th:Haan}, allows us to write a Haantjes operator $\boldsymbol{L}$ in a  block-diagonal  form.
Denote by
\begin{equation}
 \label{eq:E}
 \mathcal{E}_i := Im \Bigl(\boldsymbol{L}-l_i\mathbf{I}\Bigr)^{\rho_i }=  \bigoplus_{{j=1,\, j\neq i}} ^s  \mathcal{D}_j, \qquad\qquad i=1,\ldots,s
\end{equation}
  the  distribution of corank $r_i$ (being $r_i$ the rank of $\mathcal{D}_i$), spanned by all the generalized eigenvectors of  $\boldsymbol{L}$, except those associated with the eigenvalue $l_i$.   Such a distribution  will be said to be a \emph{characteristic distribution} of  $\boldsymbol{L}$. Let $\mathcal{E}^{\circ}_{i}$ denote the annihilator of the distribution $\mathcal{E}_{i}$. Observe that, since $\boldsymbol{L}$  by hypothesis has real eigenvalues, the   cotangent spaces of $M$ can be locally decomposed as
\begin{equation}
 \label{eq:TMdscomp}
T_{\boldsymbol{x}}^*M=\bigoplus_{i=1} ^s  \mathcal{E}_{i}^{\circ}(\boldsymbol{x}).
\end{equation}
Moreover,  each  characteristic distribution    $\mathcal{E}_i$  is  integrable by virtue of  Theorem  \ref{th:Haan}. We shall denote by $ \mathrm{E}_i$ the foliation associated  with $\mathcal{E}_i$  and by
$E_i(\boldsymbol{x})$ the connected leave through $\boldsymbol{x}$  belonging to $ \mathrm{E}_i$.  Thus,  the set  of distributions $\{\mathcal{E}_1, \mathcal{E}_2, \ldots ,\mathcal{E}_s\}$
generates  as many  foliations $\{ \mathrm{E}_1,  \mathrm{E}_2, \ldots ,  \mathrm{E}_s\}$ as the number of distinct eigenvalues of $\boldsymbol{L}$.   This set of foliations
  will be referred to as the \textit{characteristic  web} of $\boldsymbol{L}$ and the leaves $E_i(\boldsymbol{x})$ of each foliation $ \mathrm{E}_i$ as the  \emph{characteristic fibers} of the web.

\begin{definition}
A collection of $r_i$ smooth functions will be said to be adapted to a foliation $E_i$ of the  characteristic  \textit{web} of a Haantjes operator $\boldsymbol{L}$  if the level sets of such functions coincide with the characteristic fibers of the foliation $E_i$.
\end{definition}
\begin{definition}
A parametrization of the characteristic web of a Haantjes  operator $\boldsymbol{L}$   is an ordered set  of $n$ independent smooth functions grouped as
$(\boldsymbol{f}^1, \ldots,\boldsymbol{f}^i,\ldots, \boldsymbol{f}^s) $
such that for any $i=1,\ldots, s$, the ordered subset
$\boldsymbol{f}^i=(f^{i,1}, \ldots , f^{i,r_i})$ is adapted to the  $i$th characteristic foliation of the web:
\begin{equation}
\label{eq:fad}
f^{i,k}_{\vert   E_i(\mathbf{x})}=c^{i,k}  \qquad \forall E_i (\mathbf{x})\in  \mathrm{E}_i \ ,\quad
k=1,\ldots,r_i\ ,
\end{equation}
where $c^{i,k}$ are real constants depending  on the indices $i$ and $k$ only.
In this case, we shall say that the  collection of all these functions is adapted to the web and that each of them is a \emph{characteristic function}.
\end{definition}
\begin{proposition}\label{cor:Hframe}
The vanishing of the Haantjes torsion of an operator field $\boldsymbol{L}$ is a  sufficient condition to ensure that $\boldsymbol{L}$ admits an equivalence class  of integrable generalized eigen-frames, where $\boldsymbol{L}$ takes a block-diagonal form. Furthermore, if $\boldsymbol{L}$ is semisimple, the condition is also necessary and $\boldsymbol{L}$ takes a diagonal form. In addition, if $\boldsymbol{L}$ is \emph{simple} each  of its eigen-frames is integrable.
\end{proposition}
\begin{proof}
 Each characteristic distribution $\mathcal{E}_i$ is integrable, by virtue of the Haantjes Theorem \ref{th:Haan}. Therefore, in the corresponding annihilator $\mathcal{E}_i^\circ$ one can find  $r_i$ exact one-forms $(\rd x^{i,1}, \ldots, \rd x^{i,r_i})$ that provide us with  functions
$\boldsymbol{x}^{i}=(x^{i,1}, \ldots, x^{i,r_i})$ adapted to the characteristic foliation $\mathrm{E}_i$. Collecting together all these functions, one can construct a set of $n$ independent coordinates, that we group as  $(\boldsymbol{x}^1, \ldots,\boldsymbol{x}^i,\ldots, \boldsymbol{x}^s)$ and consequently, a local chart $\{U, (\boldsymbol{x}^1, \ldots, \boldsymbol{x}^i,\ldots, \boldsymbol{x}^s)\}$, adapted to the characteristic web.
The corresponding natural frame
$\left\{\frac{\partial}{\partial \boldsymbol{x}^1}, \ldots, \frac{\partial}{\partial \boldsymbol{x}^i},
\ldots,\frac{\partial}{\partial \boldsymbol{x}^s}\right\}$
turns out to be  a generalized eigen-frame. In fact, as
\begin{equation}
 \label{eq:Dann}
\mathcal{D}_i^\circ =\bigoplus_{{j=1,\, j\neq i}} ^s  \mathcal{E}_{j}^{\circ},
\end{equation}
any  generalized eigenvector $W \in \mathcal{D}_i$ leaves invariant all the coordinate functions except at most the characteristic functions $\boldsymbol{x}^i=(x^{i,1}, \ldots, x^{i,r_i})$ of $\mathrm{E}_i$.
Thus, we have that $W=W(\boldsymbol{x}^i)\frac{\partial}{\partial \boldsymbol{x}^i}= \sum_{k=1}^{r_i}  W(x^{i,k}) \frac{\partial}{\partial x^{i,k}}$, therefore
\begin{equation}
\label{eq:Dbase}
\mathcal{D}_{i_{\vert U}}=\left\langle \frac{\partial}{\partial x^{i,1}},\ldots,
\frac{\partial}{\partial x^{i,r_i}}\right\rangle
\end{equation}
(hereafter the symbol $\langle   \rangle $ denotes the $C^{\infty}(M)$-linear span of the considered vector fields). Thus, each frame \emph{equivalent} to $\left\{\frac{\partial}{\partial \boldsymbol{x}^1}, \ldots, \frac{\partial}{\partial \boldsymbol{x}^i},\ldots,\frac{\partial}{\partial \boldsymbol{x}^s}\right\}$ is an integrable eigen-frame of \textit{generalized eigenvector fields}. Consequently, there exists an equivalence class of integrable frames and associated local charts where the operator $\boldsymbol{L}$ takes a block-diagonal form due to the invariance of its eigen-distributions.
Moreover, if $\boldsymbol{L}$ is semisimple, then its generalized eigen-frames are proper eigen-frames and the block-diagonal form reduces to a diagonal one.

Conversely, if there exists a local chart where $\boldsymbol{L}$ takes a diagonal form, then the corresponding natural frame is obviously integrable.
Thus, due to Proposition \ref{pr:Hd}, the Haantjes torsion of  $\boldsymbol{L}$  vanishes.
Finally,  if $\boldsymbol{L}$ is simple, each eigen-distribution has rank $1$; therefore, each natural eigen-frame fulfills the conditions of Proposition \ref{pr:BCRframe}. The last statement is equivalent  to the celebrated result that
 A. Nijenhuis published in 1951~\cite{Nij}.
\end{proof}
\begin{definition}\label{def:HchartL}
Let $\boldsymbol{L}:TM \to TM$ be a Haantjes operator. A local chart $\{ U, (x^1,\ldots,x^n)\}$ whose natural frame is a generalized eigen-frame  will
be said to be a Haantjes chart for $\boldsymbol{L}$.
\end{definition}
A Haantjes chart  for $\boldsymbol{L}$ can also be computed by using the transposed operator  $\boldsymbol{L}^{T}$.
Let us denote by
\begin{equation}
\label{eq:geigenf}
\ker \Bigl(\boldsymbol{L}^T-l_i\boldsymbol{I}\Bigr)^{\rho_i}
\end{equation}
 the $i$th distribution of the  generalized eigen $1$-forms  associated with the eigenvalue   $l_i(\boldsymbol{x})$,
  which fulfills the property
\begin{equation}
 \label{eq:autof}
 \ker(\boldsymbol{L}^T-l_i \boldsymbol{I})^{\rho_i}={\biggl( Im {\Bigl(\boldsymbol{L}-l_i\mathbf{I}\Bigr)}^{\rho_i } \biggr)}^\circ=
{\mathcal{E}}_{i}^\circ \ .
\end{equation}
Such a property implies that each  generalized eigenform of $\boldsymbol{L}^T$  annihilates all generalized eigenvectors of $\boldsymbol{L}$ associated with different eigenvalues. Moreover, it allows us to prove an interesting equivalence result.

\begin{proposition}\label{pr:autoformeL}
 Let  $\boldsymbol{L}: TM \to TM$ be a Haantjes operator.  The  differentials   of the characteristic coordinate  functions   are   exact generalized eigenforms for the transposed operator $\boldsymbol{L}^T$.
 Conversely, each  (locally) exact  generalized eigenform of $\boldsymbol{L}  ^T  $ provides us with a   \emph{characteristic function} for the Haantjes web of $\boldsymbol{L}$.
\end{proposition}
The characteristic  functions of a Haantjes operator are characterized by the following, simple property.
\begin{proposition}\label{pr:fchar}
A function  $h\in C^{\infty}(M)$ is a characteristic function of a Haantjes operator $\boldsymbol{L}$, associated with its eigenvalue $l_i$, if and only if given a set of local coordinates adapted to the characteristic web
$(\boldsymbol{x}^1,\ldots,\boldsymbol{x}^i,\ldots,\boldsymbol{x}^s)$,
$h$ depends, at most,  on the subset of coordinates $\boldsymbol{x}^i=(x^{i,1}, \ldots, x^{i,r_i})$ that are constant over the leaves of the foliation $\mathrm{E}_i$.
\end{proposition}
\begin{proof}
 If $h=h (x^{i,1}, \ldots, x^{i,r_i})$, it is constant on the leaves of $\mathrm{E}_i$, then $\rd h \in \mathcal{E}_i^\circ$.
Vice versa, if we assume that $\rd h\in \mathcal{E}_i^\circ$, then it can be expressed in terms of a linear combination  (with functions as coefficients) of $\{\rd x^{i,1}, \ldots, \rd x^{i,r_i}\}$ only. The result follows by observing that $\rd h$ is an exact 1-form.
\end{proof}

\section{Characteristic Haantjes coordinates and block  diagonalization}

We shall present below some of our main results, concerning the existence of charts of coordinates in which a family of commuting Haantjes operators takes a block-diagonal form.
\begin{proposition}\label{th:HJ}
Let  $\mathcal{K}=\{\boldsymbol{K}_1,\ldots,\boldsymbol{K}_w\}$, $\boldsymbol{K}_\alpha: TM\to TM$, $\alpha=1,\ldots,w$ be  a family  of commuting   operator fields; we assume that  one of them, say $\boldsymbol{K}_1$, has vanishing Haantjes torsion.  Then, there exist local charts in which all of the operators $\boldsymbol{K}_\alpha$ can be written simultaneously in a block-diagonal form. In addition, if we assume that

(i) all the operators of the family have  vanishing Haantjes torsion,

(ii) all possible nontrivial intersections of their generalized eigen-distributions
\begin{equation}
\label{eq:Va}
\mathcal{V}_a(\boldsymbol{x}):=  \bigoplus_{i_1,\ldots,i_w}^{s_1,\ldots,s_w}\mathcal{D}_{i_1}^{(1)}(\boldsymbol{x}) \bigcap  \ldots      \bigcap \mathcal{D}_{i_w}^{(w)}(\boldsymbol{x}) \ ,\qquad a=1,\ldots, v\leq n
\end{equation}
are mutually integrable,
 then  there exist  sets of local coordinates, adapted to the decomposition
\begin{equation}
 \label{eq:TVa}
T_{\boldsymbol{x}}M= \bigoplus_{a=1}^{v}\mathcal{V}_a( \boldsymbol{x}) \qquad  \boldsymbol{x} \in M ,
\end{equation}
in which all operators $\boldsymbol{K}_\alpha$  admit  simultaneously a block-diagonal form (with  possibly finer blocks).
\end{proposition}
\begin{proof}
Due to Proposition \ref{cor:Hframe}, there exists an equivalence class of integrable frames and local charts where the Haantjes operator $\boldsymbol{K}_1$ takes a block-diagonal form.
Such coordinates  are adapted to the characteristic web  associated with the spectral decomposition of $\boldsymbol{K}_1$:
\begin{equation}
\label{eq:decomp}
T_{\boldsymbol{x}}M= \mathcal{D}_{i_1}^{(1)}(\boldsymbol{x})\bigoplus \mathcal{E}_{i_1}^{(1)}(x)=\bigoplus_{i_1=1}^{s_1}\mathcal{D}_{i_1}^{(1)}(\boldsymbol{x})
\end{equation}
 These coordinates will be denoted by
\begin{equation}
\label{eq:K1chart}
\boldsymbol{x}= (\boldsymbol{x}^1,\ldots,\boldsymbol{x}^{i_1}, \ldots,\boldsymbol{x}^{s_1})
\end{equation}
where $\boldsymbol{x}^{i_1}=(x^{i_1,1},\ldots,  x^{i_1,r_{i_1}} )$ are defined over the integral leaves
 of the eigen-distribution $\mathcal{D}_{i_1}^{(1)}$ and the remaining ones, namely
\[
(\boldsymbol{x}^1,\ldots, \boldsymbol{x}^{{i_1}-1}, \boldsymbol{x}^{{i_1}+1},\ldots,\boldsymbol{x}^{s_1})
\]
 are coordinates of the leaves, i.e.,~are constant
 $(\boldsymbol{x}^1=\boldsymbol{c}_1^1,\ldots, \boldsymbol{x}^{{i_1}-1}=\boldsymbol{c}_1^{{i_1}-1}, \boldsymbol{x}^{{i_1}+1}=\boldsymbol{c}_1^{{i_1}+1}, \ldots,\boldsymbol{x}^{s_1}=\boldsymbol{c}_1^{s_1})$
 on each leaf $D^{(1)}_{i_1}(\boldsymbol{c}_1)$ of the foliation. Here
 $\boldsymbol{c}_1:=(\boldsymbol{c}_1^1,\ldots, \boldsymbol{c}_1^{{i_1}-1}, \boldsymbol{c}_1^{{i_1}+1}, \ldots,\boldsymbol{c}_1^{s_1}) $ (the subindex $1$ refers to the operator $\boldsymbol{K}_1$).
 All operators of the family $\mathcal{K}$ commute; thus, every distribution $\mathcal{D}_{i_1}^{(1)}$ is left invariant by the operators $\{\boldsymbol{K}_2,\ldots,\boldsymbol{K}_w\}$. This implies that all the operators  $\boldsymbol{K}_\alpha\in \mathcal{K}$ in the local chart \eqref{eq:K1chart}  take a block-diagonal form
   \[
[\boldsymbol{K}_\alpha]=
\left[\begin{array}{c|c|c}
\boldsymbol{K}_1^{(\alpha)} &  &  \\
\hline  & \ddots &  \\
\hline  &  & \boldsymbol{K}_{s_1}^{(\alpha)}
\end{array}\right] \ , \qquad \alpha=1,\ldots,w
\]
\unskip\ignorespaces where $[\boldsymbol{K}^{(\alpha)}_{i_1}]_{ \vert _{D^{(1)}_{i_1}(\boldsymbol{c}_1)}}=[\boldsymbol{K}_{\alpha \vert _{D^{(1)}_{i_1}(\boldsymbol{c}_1)}}]$, $i_1=1,\ldots,s_1$, since the operators
 $\boldsymbol{K}_\alpha$ can be restricted  to each leaf $D^{(1)}_{i_1}(\boldsymbol{c}_1)$.

Besides, if any other operator of the set $\{\boldsymbol{K}_2,\ldots,\boldsymbol{K}_w\}$, say
$\boldsymbol{K}_2$, has vanishing Haantjes torsion, then the tangent space at any point $\boldsymbol{x} $ admits the finer decomposition
\begin{equation}
 \label{eq:TMinters12}
T_{\boldsymbol{x}}M= \bigoplus_{i_1,i_2}^{s_1,s_2}\mathcal{D}_{i_1}^{(1)}(\boldsymbol{x}) \cap \mathcal{D}_{i_2}^{(2)}(\boldsymbol{x}) \ ,
\end{equation}
where $\mathcal{D}_{i_2}^{(2)}$ are  the generalized eigen-distributions    of $\boldsymbol{K}_2$, which are integrable
by virtue of the Haantjes Theorem \ref{th:Haan}.
Consequently, the Haantjes Theorem can also be applied to the restriction of $\boldsymbol{K}_2$ to  $D^{(1)}_{i_1}(\boldsymbol{c}_1)$. Therefore, there exists a transformation of coordinates, acting only on the coordinates over the leaves of the foliation $D^{(1)}_{i_1}$
\begin{equation}
\Phi: M\rightarrow M, \qquad
(\boldsymbol{x}^1,\ldots,\boldsymbol{x}^{i_1}, \ldots,\boldsymbol{x}^{s_1}) \mapsto
(\boldsymbol{x}^1,\ldots,\boldsymbol{y}^{i_1}, \ldots,\boldsymbol{x}^{s_1}) \ ,
\end{equation}
such that the new coordinates $\boldsymbol{y}^{i_1}= (y^{i_1, 1},\ldots,y^{i_1 ,r_{i_1}})=\boldsymbol{f}^{i_1}(\boldsymbol{x}^{i_1})$ are
adapted to the decomposition
\begin{equation}
 \label{eq:TMinters12}
T_{\boldsymbol{x}}D^{(1)}_{i_1}(\boldsymbol{c}_1)= \bigoplus_{i_2}^{s_2}\mathcal{D}_{i_1}^{(1)}(\boldsymbol{x}) \cap \mathcal{D}_{i_2}^{(2)}(\boldsymbol{x}), \qquad \boldsymbol{x} \in
D^{(1)}_{i_1}(\boldsymbol{c}_1)\ .
\end{equation}
 Thus, we have
    \begin{equation}
\label{eq:K2diag}
 [\boldsymbol{K}^{(\alpha)}_{i_1}]=
\left[\begin{array}{c|c|c}
\boldsymbol{K}_{i_1,1}^{(\alpha)} & 0 & 0 \\
\hline 0 & \ddots & 0 \\
\hline 0 & 0 & \boldsymbol{K}_{i_1,s_2}^{(\alpha)}
\end{array}\right] \ , \qquad \alpha=1, \dots, w
\end{equation}
 where $[\boldsymbol{K}_{i_1,j}^{(\alpha)}]_
 { \vert _{D_{i_1}^{(1)}(\boldsymbol{c}_1) \cap D_{j}^{(2)}(\boldsymbol{c}_2)}} =\left[\boldsymbol{K}_
 {\alpha \vert _{D_{i_1}^{(1)}(\boldsymbol{c}_1) \cap D_{j}^{(2)}(\boldsymbol{c}_2)}} \right]$,
 $j=1,\ldots,s_2$.
  Let us consider  the decomposition
\begin{equation}
 \label{eq:TMinters12}
T_{\boldsymbol{x}}M= \bigoplus_{i_1,i_2}^{s_1,s_2}\mathcal{D}_{i_1}^{(1)}(\boldsymbol{x}) \cap \mathcal{D}_{i_2}^{(2)}(\boldsymbol{x})=
\bigoplus_{\gamma=1}^{u}\mathcal{U}_\gamma (\boldsymbol{x})\qquad \boldsymbol{x} \in M\ ,
\end{equation}
where in the direct sum \eqref{eq:TMinters12}  $\mathcal{U}_\gamma \neq\{\boldsymbol{0}\}$, $u\leq n$ and $r_\gamma$ denotes the rank of  $\mathcal{U}_\gamma$ ($\sum_{\gamma=1}^u r_\gamma=n$). Clearly, the  distributions  $\mathcal{U}_\gamma$ are invariant under the action of each operator $\boldsymbol{K}_\alpha\in\mathcal{K}$. Besides, these distributions are involutive, and are realized as the intersection of involutive distributions. By assumption, they are also \emph{mutually} integrable; therefore, there exist  local charts in $M$ of the form
\begin{equation}
\label{eq:diag12}
\{ U ,(y^{1,1},\ldots,y^{1,r_{1}}; \dots;y^{i_1, 1},\ldots,y^{i_1 ,r_{i_1}};\ldots ;y^{s_{1},1},\dots,y^{s_{1},r_{s_1}})\}
\end{equation}
adapted to the decomposition \eqref{eq:TMinters12}, where all the operators $\boldsymbol{K}_\alpha\in\mathcal{K}$  admit simultaneously a (possibly) finer block-diagonal form.
By extending the previous procedure to the Haantjes operators $\boldsymbol{K}_3, \ldots, \boldsymbol{K}_w$, we obtain the decomposition
\begin{equation}
 \label{eq:TMinters}
T_{\boldsymbol{x}}M= \bigoplus_{i_1,\ldots,i_w}^{s_1,\ldots,s_w}\mathcal{D}_{i_1}^{(1)}(\boldsymbol{x}) \bigcap  \ldots      \bigcap \mathcal{D}_{i_w}^{(w)}(\boldsymbol{x})= \bigoplus_{a=1}^{v}\mathcal{V}_a(  \boldsymbol{x})  ,
\end{equation}
where in the  direct sum \eqref{eq:TMinters} $\mathcal{V}_a \neq \{\boldsymbol{0}\}$, $v\leq n$ and $r_a$ denotes the rank of $\mathcal{V}_a$  ($\sum_{a=1}^v r_a=n$).
Then, as the involutive distributions $\mathcal{V}_a$ by assumption are \emph{mutually} integrable,  there exist  local charts
\begin{equation}
\label{eq:HchartK}
\{ U,(\boldsymbol{y}^{1},\ldots , \boldsymbol{y}^{a},\ldots,   \boldsymbol{y}^{v}) \} \ ,
\end{equation}
adapted to the decomposition \eqref{eq:TMinters}, such that
\begin{equation}
 \label{eq:Va}
\mathcal{V}_a=\left \langle  \frac{\partial}{\partial y^{a, 1}},\ldots , \frac{\partial}{\partial y^{a, r_a}} \right  \rangle  \qquad a=1,\ldots,v,
\end{equation}
where the natural frame $\big\{ \frac{\partial}{\partial y^{a, 1}},\ldots , \frac{\partial}{\partial y^{a, r_a}} \big\}$ over the leaves of $\mathcal{V}_a$ is formed by joint \emph{generalized} eigenvector fields of the operators $\{\boldsymbol{K}_1,\ldots,\boldsymbol{K}_w\}$.
\end{proof}
\begin{definition}\label{def:HchartK}
Let   $\mathcal{K}=\{\boldsymbol{K}_1,\ldots,\boldsymbol{K}_w\}$ be  a family  of   Haantjes operators. A \textit{Haantjes chart} for the family $\mathcal{K}$ is defined to be  a local chart $\{ U, (x^1,\ldots,x^n)\}$ whose associated natural frame is a generalized eigen-frame of each $\boldsymbol{K}_\alpha\in \mathcal{K}$.
\end{definition}
The assumption that the generalized eigen-distributions of the Haantjes operators in $\mathcal{K}$ are mutually integrable  is sufficient to guarantee the existence of Haantjes charts for the family $\mathcal{K}$. In the next section, we will show that in the case of  semisimple Haantjes  algebras, such hypothesis is always satisfied.

To conclude this analysis, we prove a result that will be useful for the proof of Theorem \ref{th:HaanA} of the next section.
\begin{proposition}\label{pr:HaSomme}
Let $\boldsymbol{K}_1$ and $\boldsymbol{K}_2:TM \to TM$ be two commuting   operator fields, and $\mathcal{D}_i^{(1)}$, $  \mathcal{D}_j^{(2)}$, $i\in \{1,\ldots,s_1\}$, $j\in\{1,\ldots,s_2\}$ two   involutive generalized eigen-distributions of $\boldsymbol{K}_1$ and $\boldsymbol{K}_2$, with Riesz indices $\rho_i$ and $\sigma_j$, respectively, which are assumed to be independent of $\boldsymbol{ x}$. The following three conditions are equivalent:
\begin{enumerate}
\item
the distribution $\mathcal{D}_i^{(1)} + \mathcal{D}_j^{(2)}$ is involutive;
\item
${\tau}_{(\boldsymbol{K}_1-l_i ^{(1)}\boldsymbol{I})^{\rho_i}(\boldsymbol{K}_2-l_j ^{(2)}\boldsymbol{I})^{\sigma_j}}
(\mathcal{D}_i ^{(1)} + \mathcal{D}_j^{(2)}, \mathcal{D}_i^{(1)}  + \mathcal{D}_j^{(2)})=\mathbf{0} \ ;
$
\item
$
\mathcal{H}_{(\boldsymbol{K}_1-l_i^{(1)}\boldsymbol{I})^{\rho_i}(\boldsymbol{K}_2-l_j ^{(2)}\boldsymbol{I})^{\sigma_j}}
(\mathcal{D}_i^{(1)}  + \mathcal{D}_j^{(2)}, \mathcal{D}_i ^{(1)} + \mathcal{D}_j^{(2)})=\mathbf{0}  \ .
$
\end{enumerate}

\end{proposition}
\begin{proof}
The proof proceeds as in Theorem \ref{pr:FXmn}, with the identification $\boldsymbol{A}:=  (\boldsymbol{K}_1-l_i^{(1)} \boldsymbol{I})^{\rho_i}$,
$\boldsymbol{B}:=  (\boldsymbol{K}_2-l_j ^{(2)}\boldsymbol{I})^{\sigma_j}$ and exploiting Eqs.~\eqref{eq:Riesz}, \eqref{eq:rankAB} and \eqref{eq:cor6}.
\end{proof}
The latter proposition can be easily extended to an arbitrary number of commuting operators.

\section{Haantjes algebras}

In this section, we shall introduce the notion of \textit{Haantjes algebra}, which  is essentially an associative algebra of operator fields with
vanishing Haantjes torsion.

\subsection{Main Definition}

\begin{definition}\label{def:HM}
A Haantjes algebra of rank $m$ is a pair    $(M, \mathscr{H})$ with the following properties:
\begin{itemize}
\item
$M$ is a differentiable manifold of dimension $\mathrm{n}$;
\item
$ \mathscr{H}$ is a set of Haantjes  operators $\boldsymbol{K}:TM\rightarrow TM$   that  generate
\begin{itemize}
\item
a free module of rank $\mathrm{m}$   over the ring of smooth functions on $M$:
\begin{equation}
\label{eq:Hmod}
\mathcal{H}_{\bigl( f\boldsymbol{K}_{1} +
                             g\boldsymbol{K}_2\bigr)}(X,Y)= \mathbf{0}
 \ , \qquad\forall\, X, Y \in TM \ , \quad \, f,g \in C^\infty(M)\  ,\quad \forall ~\boldsymbol{K}_1,\boldsymbol{K}_2 \in  \mathscr{H};
\end{equation}
  \item
a ring  w.r.t. the composition operation
\begin{equation}
 \label{eq:Hring}
\mathcal{H}_{\bigl(\boldsymbol{K}_1 \, \boldsymbol{K}_2\bigr)}(X,Y)=
 \mathcal{H}_{\bigl(\boldsymbol{K}_2 \, \boldsymbol{K}_1\bigr)}(X,Y)=\mathbf{0} \ , \qquad
\forall\, \boldsymbol{K}_1,\boldsymbol{K}_2\in  \mathscr{H} , \quad\forall\, X, Y \in TM\ .
\end{equation}
\end{itemize}
\end{itemize}
In addition, if
\begin{equation}
\boldsymbol{K}_1\,\boldsymbol{K}_2=\boldsymbol{K}_2\,\boldsymbol{K}_1 \ , \quad\qquad\ \boldsymbol{K}_1,\boldsymbol{K}_2 \in  \mathscr{H}\ ,
\end{equation}
the  algebra $(M, \mathscr{H})$ will be said to be an Abelian Haantjes algebra. Moreover, if   the identity operator $\boldsymbol{I}\in \mathscr{H}$, then $(M, \mathscr{H})$ will be said to be a Haantjes algebra with identity.
\end{definition}
The assumptions \eqref{eq:Hmod}, \eqref{eq:Hring} ensure that the set $\mathscr{H}$ generates an associative algebra of Haantjes operators. Besides, whenever $\boldsymbol{K} \in \mathscr{H}$, then the powers $\boldsymbol{K}^i \in \mathscr{H}  ~ \forall\, i\in \mathbb{N} \backslash \{0\}$.

Let us consider the minimal polynomial of an operator  $\boldsymbol{K} \in \mathscr{H}$
\begin{equation}
\label{eq:mpK}
 m_{\boldsymbol{K}}(\boldsymbol{x},\lambda)=\prod_{i=1}^s \bigl(\lambda-l_i(\boldsymbol{x})\bigr)^{\rho_i}=\lambda^r+\sum_{j=1} ^{r} c_{j}(\boldsymbol{x}) \lambda^{r-j }\ ,\qquad r=\sum_{i=1}^s\rho_i \ ,
\end{equation}
 where $l_i (\boldsymbol{x})$, $i=1,\ldots,s$ are the pointwise distinct eigenvalues of  $\boldsymbol{K}$. The previous observations imply the following results.
\begin{lemma}\label{pr:Hg}
Let $(M, \mathscr{H})$ be a Haantjes algebra of rank $\mathrm{m}$. If $\boldsymbol{I}\in \mathscr{H}$, then the degree $\mathrm{r}$ of the minimal polynomial of each $\boldsymbol{K} \in \mathscr{H}$ is not greater than $\mathrm{m}$;  if  $\boldsymbol{I}\notin \mathscr{H}$, then $\mathrm{r}\leq (\mathrm{m}+1)$.
\end{lemma}

\begin{lemma}\label{lm:HnotI}
Let $(M, \mathscr{H})$ be a Haantjes algebra without identity. Then, each  $\boldsymbol{K} \in \mathscr{H}$ is a non-invertible operator.
\end{lemma}
\begin{proof}
Let $m_{\boldsymbol{K}}(\boldsymbol{x}, \lambda)$  be the minimal polynomial of $\boldsymbol{K}$, as in Eq.~\eqref{eq:mpK}. Let us consider the operator $c_r(\boldsymbol{x}) \boldsymbol{I}=-( \boldsymbol{K}^r+\sum_{j=1} ^{r-1} c_{j}(\boldsymbol{x}) \boldsymbol{K}^{r-j } )$.  At the points  $\boldsymbol{x}\in M$ where $c_r(\boldsymbol{x})\neq 0$,  this operator belongs to $\mathscr{H}$. Then, on this set of points, the identity operator $\boldsymbol{I}$ should also belong to $\mathscr{H}$, which  is absurd. Therefore $c_r(\boldsymbol{x})=(-1)^r\ \Pi_{i=1}^sl_i ^{\rho_i}(\boldsymbol{x}) $  vanishes $\forall \,\boldsymbol{x}\in M$. Consequently,  $\boldsymbol{K}$ is not invertible at any point of $M$.
\end{proof}
\begin{remark}
The class of  $\omega\mathcal{H}$ manifolds introduced and discussed in~\cite{TT2016prepr} can be regarded as a family of symplectic manifolds of dimension $2n$, endowed with an Abelian
Haantjes algebra of rank $m=n$ that fulfills an additional compatibility condition with the symplectic form $\omega$.
\end{remark}
The requirements of Definition \ref{def:HM} are apparently demanding and difficult to fulfill in concrete examples. However, a natural Haantjes algebra  is realized, in a local chart
$\{U, \boldsymbol{x}=(x^1,\ldots,x^n)\}$, by  any set of diagonal operators of the form
\begin{equation}
\label{eq:Hdiag}
\boldsymbol{K}=\sum _{k=1}^n l_{k }(\boldsymbol{x})
\frac{\partial}{\partial x^k}\otimes \rd x^k  \ ,
\end{equation}
 where the smooth functions $ l_{k }(\boldsymbol{x})$ play the role of eigenvalue fields of $\boldsymbol{K}$.
 The   operators of the form \eqref{eq:Hdiag} have  vanishing  Haantjes torsion and satisfy the differential compatibility condition \eqref{eq:Hmod}, by virtue of Proposition \ref{pr:Hd}.  Moreover, they form
a commutative ring, as they also satisfy the relations \eqref{eq:Hring}.
 In fact, such operators generate an algebraic structure that will be said to be a \emph{diagonal} Haantjes algebra.
In the following section, we will show that this interesting structure represents the prototype of Haantjes algebras of commuting \emph{semisimple} Haantjes operators. This observation suggests the following

\begin{definition}\label{def:Hdiag}
A  Haantjes algebra $(M, \mathscr{H})$ is said to be semisimple if each operator $\boldsymbol{K} \in \mathscr{H}$ is semisimple.
\end{definition}
We prove now a general result concerning the intersections of  the kernels of  semisimple operator fields, which will be crucial  in the proof of Theorem \ref{th:HaanA}.
\begin{lemma}\label{lm:kerIntersections}
Let $\{\boldsymbol{L}_1,\ldots,\boldsymbol{L}_w\}$, $\boldsymbol{L}_\alpha: TM \to TM$, $\alpha=1,\ldots,w$ be semisimple (non-invertible) commuting operator fields. Then, there exist $w$ functions $(c_1,\ldots,c_w)$ such that
\begin{equation}
 \label{eq:kerIntersections}
\bigcap _{i=1}^w \ker \boldsymbol{L}_i=\ker \left(\sum_{i=1}^w c_i \boldsymbol{L}_i\right )
\end{equation}
\end{lemma}
\begin{proof}
 The operators $\{\boldsymbol{L}_1,\ldots,\boldsymbol{L}_w\}$ are semisimple and commute. Therefore there exists a local reference frame where they
are, simultaneously, pointwise diagonalizable.
Without loss of generality, we can suppose
that $\{\boldsymbol{L}_2,\ldots,\boldsymbol{L}_s\}$, with $s\le w$, have kernels whose intersections with $\ker \boldsymbol{L}_1$ have strictly decreasing rank. Then, we can choose a local reference frame of common eigenvector fields
$ \{X_1,\ldots,X_n\}$ in such a way that the eigenvalue fields of $\{\boldsymbol{L}_1,\ldots,\boldsymbol{L}_s\}$, at each $ \boldsymbol{x}\in M$ fulfill the conditions
\begin{eqnarray*}
l_j^{(1)}( \boldsymbol{x})&\neq &0 \ , \qquad j=1,\ldots, r_1 \ ,\\
l_j^{(2)}( \boldsymbol{x})&\neq &0 \ , \qquad j=(r_1+1),\ldots, r_2 \ , \\
\ldots &\neq &0\ldots \ ,\\
l_j^{(s)}( \boldsymbol{x})&\neq &0 \ , \qquad j=(r_{s-1}+1),\ldots, r_s \ ,
\end{eqnarray*}
where $r_1=\rank (\boldsymbol{L}_1), \ldots , r_i=  \rank(\boldsymbol{L}_i)_{\vert \ker\boldsymbol{L}_{i-1}}$  for $i=2,\ldots,s$. Then, we can choose $w$ functions $c_1,\ldots,c_w$ such that $\forall \, \boldsymbol{x}\in M$ :
\begin{eqnarray}
c_j( \boldsymbol{x})&=&0 \ ,\qquad j=s+1,\ldots,w \ , \\
c_s( \boldsymbol{x})&\neq& 0 \ ,\\
c_{s-1}( \boldsymbol{x})&\neq&- \frac{1}{l_j^{(s-1)}( \boldsymbol{x})}c_s( \boldsymbol{x}) l_j^{(s)}(\boldsymbol{x}) \ ,\qquad  j=(r_{s-2}+1),\ldots,r_{s-1} \ , \\
\ldots& &\ldots\\
c_1( \boldsymbol{x})&\neq&- \frac{1}{l_j^{(1)}( \boldsymbol{x})}\sum_{k=2}^s c_k ( \boldsymbol{x})l_j^{(k)}(\boldsymbol{x}) \ ,\qquad j=1,\ldots,r_1\ .
\end{eqnarray}
 With such choices, the linear combination $\sum_{i=1}^w c_i \boldsymbol{L}_i$ satisfies the property \eqref{eq:kerIntersections}.
\end{proof}
We are now able to prove a relevant result concerning the integrability of the eigen-distributions of a set of Haantjes operators forming a semisimple, Abelian Haantjes algebra.
\begin{theorem}\label{th:HaanA}
Let $(M, \mathscr{H})$ be a semisimple Abelian Haantjes algebra, and assume that the rank of the  eigen-distributions of the operators belonging to $\mathscr{H}$ is independent of $\boldsymbol{x}\in M$.
Then, these   eigen-distributions, jointly with all of their possible intersections, are mutually integrable.
\end{theorem}
\begin{proof}
Let $\{\boldsymbol{K}_1,\ldots,\boldsymbol{K}_m\}$ be a basis of $(M, \mathscr{H})$, and
$ (\mathcal{D}_{i_1}^{(1)},  \ldots    ,\mathcal{D}_{i_m}^{(m)})$, $i_1=1,\ldots, s_1$, $i_m=1,\ldots,s_m$, be the set of their (proper) eigen-distributions. Let
\begin{equation}
\label{eq:1ints}
 \mathcal{V}_a=\mathcal{D}_{i_1}^{(1)}(\boldsymbol{x}) \bigcap  \ldots      \bigcap \mathcal{D}_{i_m}^{(m)}(\boldsymbol{x})\qquad a=1,\ldots, v, \qquad v\leq n
\end{equation}
denote a nontrivial intersection of eigen-distributions of the operators $\{\boldsymbol{K}_1,\ldots,\boldsymbol{K}_m\}$. Consequently,  this distribution,  being the  intersection  of distributions which are involutive  (due to the Haantjes Theorem), is also involutive.  Moreover, since  $\mathcal{V}_a$ is   an intersection of the kernels of the operators $(\boldsymbol{K}_j-l_{i_j} ^{(j)}\boldsymbol{I}) , j=1,\ldots,m$, using Lemma \ref{lm:kerIntersections} it can be expressed as the kernel of a single Haantjes operator,  obtained as a  linear combination of $\{\boldsymbol{K}_1,\ldots,\boldsymbol{K}_m\}$. Then, the direct  sums of the distributions \eqref{eq:1ints}  fulfill the condition $(3)$ of Proposition \ref{pr:HaSomme}, as  $(M, \mathscr{H})$ is a  Haantjes    algebra. Therefore they  are mutually integrable.
\end{proof}
\begin{theorem}\label{TH}
Let $\mathcal{K}$ be  a family of semisimple commuting operator fields.  A set of local coordinates exists such that all the operators of $\mathcal{K}$ can be simultaneously diagonalized if and only if these operators generate an Abelian Haantjes algebra.
\end{theorem}
\label{th:DiagAll}
\begin{proof}
Assume that the operators of $\mathcal{K}$ generate an Abelian Haantjes algebra $(M, \mathscr{H})$. The result  follows from Theorem \ref{th:HaanA} and Proposition \ref{th:HJ}  by simply observing that the distributions
\eqref{eq:1ints}, by construction, are formed by proper, joint eigenvector fields of a  given  basis $\{\boldsymbol{K}_1,\ldots,\boldsymbol{K}_m\}$ of $\mathscr{H}$.
Therefore,  the natural frame $\big\{ \frac{\partial}{\partial y^{a, 1}},\ldots , \frac{\partial}{\partial y^{a, r_a}}\big\}$ over the leaves of $\mathcal{V}_a$, associated to the chart \eqref{eq:HchartK}, is formed by joint eigenvector fields of the operators  $\{\boldsymbol{K}_1,\ldots,\boldsymbol{K}_m\}$.

The converse statement can be easily proved by observing that if $\mathcal{K}$ is a family of  operators admitting local coordinates in which they all take a diagonal form,  they are necessarily  Haantjes operators and generate a \emph{diagonal} Haantjes algebra.
\end{proof}
In Section \ref{sec:cHa}, we shall present the large class of \textit{cyclic} Haantjes algebras, which are generated by the powers of a single Haantjes operator.
Examples of non-diagonal Haantjes algebras have been studied in~\cite{TT2019prepr}.

\subsection{Cyclic Haantjes algebras}
  \label{sec:cHa}
An especially relevant class of Haantjes algebras  is represented by the algebras generated by  a \textit{single} Haantjes operator   $\boldsymbol{L}:TM\mapsto TM$. In fact,  one can
consider the   algebra $\mathcal{L}$ of all powers of $\boldsymbol{L}$
\[  \label{algebraic}
\mathcal{L}(\boldsymbol{L}) := \left \langle  \boldsymbol{I} , \boldsymbol{L},\boldsymbol{L}^2, \ldots, \boldsymbol{L}^{n-1}, \ldots\right  \rangle  \ .
\]
 Due to Proposition \ref{pr:Lpowers}, the algebra $\mathcal{L}(\boldsymbol{L})$ is an example of a Haantjes algebra; we shall call it  \emph{cyclic} Haantjes algebra. Its rank is equal to the degree of the minimal polynomial of $\boldsymbol{L}$, which  is not greater than $n$, due to the Cayley--Hamilton theorem.

A natural question is to establish whether  a given Haantjes algebra  can be generated by a single Haantjes  operator. To this aim, we introduce the notion of cyclic generator of $\mathscr{H}$.
\begin{definition}\label{def:CHa}
Let  $(M,  \mathscr{H})$ be  a  rank $\mathrm{m}$ Abelian Haantjes algebra. An  operator $\boldsymbol{L}:TM \to TM$ whose minimal polynomial is of degree $\mathrm{h}\geq \mathrm{m}$  will be said to be  a cyclic generator of  $\mathscr{H}$  if
\[
 \mathscr{H}\subseteq\mathcal{L}(\boldsymbol{L}) \ .
\]
Let
\begin{equation}
 \label{eq:baseCicl}
 \mathcal{B}_{cyc}=\{ \boldsymbol{I} , \boldsymbol{L},\boldsymbol{L}^2, \ldots ,\boldsymbol{L}^{h-1} \} \ ,
\end{equation}
be   a cyclic basis of $\mathcal{L}(\boldsymbol{L})$.  A basis $\mathcal{B}$ of  $ \mathscr{H}$ such that $\mathcal{B}\subseteq \mathcal{B}_{cyc}$ will be said to be a cyclic basis of $\mathscr{H}$.
\end{definition}
A cyclic basis allows us to represent
any Haantjes operator $\boldsymbol{K}\in \mathscr{H} $  as a polynomial field in $\boldsymbol{L}$ of degree at most $(h-1)$, i.e.~\begin{equation}
\label{eq:Hg}
\boldsymbol{K}=p_{\boldsymbol{K} }(\boldsymbol{x},\boldsymbol{L})=\sum_{k =0} ^{h-1} a_k(\boldsymbol{x})\,\boldsymbol{L}^k \ ,
\end{equation}
where $a_k(\boldsymbol{x})\in C^{\infty}(M)$.
We propose now a characterization of semisimple operators, with real eigenvalues, belonging to a cyclic  algebra.
\begin{proposition}\label{pr:CGK}
Let  $\boldsymbol{L}: TM \to TM$ be a semisimple
operator with $h$ pointwise distinct eigenvalues $\{\lambda_{1}(\boldsymbol{x}), \ldots , \lambda_{h}(\boldsymbol{x})\}$. Let $\boldsymbol{K}: TM \to TM$ be another operator field possessing $s$ pointwise distinct eigenvalues, with $\text{s} \leq \text{h}$. The following conditions are equivalent:
\begin{itemize}
\item
 $\boldsymbol{K}$ belongs to the cyclic algebra of rank $\mathrm{h}$ generated by $\boldsymbol{L}$, i.e.~\begin{equation}
 \label{eq:KinL}
 \boldsymbol{K}\in \mathcal{L}(\boldsymbol{L}) \ ;
\end{equation}
\item
 there exists a polynomial field  $p_{\boldsymbol{K} }(\boldsymbol{x},\lambda)$  in $\lambda$ of degree at most $(h-1)$, such that
\begin{equation}
 \label{eq:KpL}
 \boldsymbol{K}=p_{\boldsymbol{K} }(\boldsymbol{x},\boldsymbol{L}) \ ;
\end{equation}
\item
each eigen-distribution of  $\boldsymbol{L}$ is included in a single eigen-distribution of  $\boldsymbol{K}$:
\begin{equation}
 \label{eq:cPayne}
\mathcal{C}_{\lambda_i} :=  \ker(\boldsymbol{L}-\lambda_i \boldsymbol{I})
\subseteq\mathcal{D}_{l_i} :=  \ker(\boldsymbol{K}-l_i \boldsymbol{I}),
\end{equation}
where it is understood that the eigenvalues $\Bigl(l_1(\boldsymbol{x}),\ldots,l_h(\boldsymbol{x})\Bigr)$ of $\boldsymbol{K}$ may not be all distinct.
\end{itemize}
\end{proposition}
\begin{proof}
The equivalence between conditions  \eqref{eq:KinL} and \eqref{eq:KpL} holds since the minimal polynomial of $\boldsymbol{L}$ is given by
\[
m_{\boldsymbol{L}}(\boldsymbol{x},\lambda)=\Pi_{i=1}^h \bigl(\lambda-\lambda_i(\boldsymbol{x})\bigr) \ ,
\]
and the set \eqref{eq:baseCicl} is a basis of $\mathcal{L}(\boldsymbol{L})$. Besides, condition \eqref{eq:KpL} implies \eqref{eq:cPayne} as  every eigenvector field $X$ of $\boldsymbol{L}$ belonging to $\mathcal{C}_{\lambda_i}$ is also an eigenvector field of $ \boldsymbol{K}$, with eigenvalue $l_i( \boldsymbol{x})=p_{\boldsymbol{K} }(\boldsymbol{x},\lambda_i)$:
\[
 \boldsymbol{K} X=p_{\boldsymbol{K} }(\boldsymbol{x},\boldsymbol{L}) X=p_{\boldsymbol{K} }(\boldsymbol{x},\lambda_i)X \ .
\]

Vice versa, assuming that condition \eqref{eq:cPayne} is satisfied, it suffices to show that there exists   a polynomial  $p_{\boldsymbol{K}}(\boldsymbol{x},\lambda)$, $\lambda\in \mathbb{R}$, such that $\boldsymbol{K}$ and $p_{\boldsymbol{K}}(\boldsymbol{x},\boldsymbol{L})$ agree on a basis adapted to the decomposition
\[
T_{\boldsymbol{x}}M= \bigoplus_{i=1}^{h}\mathcal{C}_{\lambda_i}(\boldsymbol{x}) \ .
\]
To this aim, we must solve the following system
\begin{equation}
\label{eq:LagPoly}
l_{i }(\boldsymbol{x})=p_{\boldsymbol{K}} (\boldsymbol{x},\lambda_i)=\sum_{k =0} ^{h-1} a_{k}(\boldsymbol{x})\lambda_{{i }}^k(\boldsymbol{x}) \ ,
\qquad\  \quad i=1,\ldots,h \ .
\end{equation}
This  can be done by means of  the $h$ Lagrange interpolation polynomials of degree $(h-1)$
\[
\pi_i(\lambda)= \frac{\Pi_{j\neq i}^h(\lambda-\lambda_j )}{\Pi_{j\neq i}^h(\lambda_i-\lambda_j )} \ ,\qquad\qquad  i=1,\ldots,h
\]
which yield the expressions
\[
p_{\boldsymbol{K}}(\boldsymbol{x},\lambda)=\sum_{i=1}^h l_{i }(\boldsymbol{x})\,\pi_i(\lambda) \ .
\]
Therefore, we obtain the relation
\begin{equation}
 \label{eq:KpK}
\boldsymbol{K}=\sum _{i=1}^h p_{\boldsymbol{K}}(\boldsymbol{x},\lambda_{i } ) \, \pi_i(\boldsymbol{L})
\end{equation}
where $ \pi_i(\boldsymbol{L})$ are pointwise projection operator fields onto the distribution $\mathcal{C}_{\lambda_i}$. The basis
$\mathcal{B}_{int}=\{\pi_1(\boldsymbol{L}),\ldots ,\pi_h(\boldsymbol{L})\}$ will be said to be a Lagrange interpolation basis of $\mathcal{L}(\boldsymbol{L})$.
\end{proof}
\begin{corollary}
If one of the equivalent conditions of Proposition \ref{pr:CGK} is satisfied and $\boldsymbol{L}$ is a Haantjes operator, then the  operator $\boldsymbol{K}$   is also a Haantjes operator which commutes with $\boldsymbol{L}$. In addition, every integrable eigen-frame of $\boldsymbol{L}$ is also an integrable eigen-frame for $\boldsymbol{K}$. Therefore we have
\begin{equation}
 \label{eq:LKred}
\boldsymbol{L}_{\vert _{C_j(\boldsymbol{x})}}= \lambda_j \boldsymbol{I}_{\vert _{C_j(\boldsymbol{x})}}\ ,\quad \qquad
\boldsymbol{K}_{ \vert _{C_j(\boldsymbol{x})}}=
p_{\boldsymbol{K}} (\boldsymbol{x},\lambda_j)\boldsymbol{I}_{ \vert _{C_j(\boldsymbol{x})}} \ ,
\end{equation}
where $C_j(\boldsymbol{x})$ denotes an integral leaf, passing through $\boldsymbol{x}$, of the eigen-distribution $\mathcal{C}_{\lambda_j}$ of $\boldsymbol{L}$.
\end{corollary}
Thus, given a semisimple Haantjes operator $\boldsymbol{K}$ with $s$ pointwise distinct eigenvalues
$\{l_1(\boldsymbol{x}),\ldots,l_s(\boldsymbol{x})\}$, one can always construct another Haantjes operator $\boldsymbol{L}$ fulfilling the condition \eqref{eq:KpL}, by considering a finer (or at least no coarser) decomposition than the spectral decomposition of $\boldsymbol{K}$
\[
T_{\boldsymbol{x}}M= \bigoplus_{i=1}^{s}\mathcal{D}_{l_i}(\boldsymbol{x}) \ ,
\]
according to which
\[
\boldsymbol{K}=\sum_{i=1}^s l_i \boldsymbol{P}_{i}(\boldsymbol{K}) \ .
\]
Here $\boldsymbol{P}_{i}$ are pointwise projection operators onto $\mathcal{D}_{l_i}$.
 To this aim, let us consider a local chart $\{U, (\boldsymbol{x}^1,\ldots, \boldsymbol{x}^s)\}$, adapted to the characteristic web of $\boldsymbol{K}$; thus, it fulfills Eq.~\eqref{eq:Dbase}.
By way of an example, one can consider the further decomposition
\begin{equation}
 \label{eq:finer}
\mathcal{D}_{l_i} (\boldsymbol{x}) = \bigoplus_{j_i=1}^{r_i}\mathcal{C}_{i,j_i}(\boldsymbol{x})\ , \qquad r_{i}= rank \hspace{1mm} \mathcal{D}_{l_i}
\end{equation}
in terms of \textit{one-dimensional} Lie subalgebras
\[
\mathcal{C}_{i,j_i}=
\left\langle \frac{\partial}{\partial x^{i,j_i}}\right\rangle
 \subseteq\mathcal{D}_{l_i}\ .
\]
 Then, one can construct the operator
\[
\boldsymbol{L}=\sum_{i=1}^{s} \lambda_{i,j_i}\sum_{j_i=1}^{ r_i}\frac{\partial}{\partial x^{i,j_i}} = \sum_{i=1,j_i=1}^{s,r_i} \lambda_{i,j_i} \,\boldsymbol{\pi}_{i,j_i}\ ,
\]
where $\boldsymbol{\pi}_{i,j_i}$ are  projection operators onto the subalgebras $\mathcal{C}_{i,j_i}$, and $\lambda_{i, j_i}$ are  arbitrarily chosen (but pointwise distinct) functions, numbered
according to the finer decomposition \eqref{eq:finer} of $T_{\boldsymbol{x}}M$. These functions will play the role of eigenvalues of $\boldsymbol{L}$. Consequently, we have
\[
\boldsymbol{P}_i=\sum_{j_i=1}^{r_i}\boldsymbol{\pi}_{i,j_i} \ .
\]
Thus, $p(\lambda_{i,j_i})=l_i$, $i=1,\ldots,s$, $j_i=1,\ldots, r_i$.

 Finally, as a consequence of Proposition~\ref{pr:CGK},  we have the following result.

\begin{proposition}\label{pr:CGKcoor}
Every    semisimple Abelian Haantjes algebra   $(M,\mathcal{H})$  is  cyclic  and admits a  cyclic Nijenhuis generator.
\end{proposition}
\begin{proof}
Let us consider the Haantjes chart
\begin{equation}
\label{eq:HchartKc}
\{ U, \boldsymbol{x}=(\boldsymbol{y}^{1}, \ldots,\boldsymbol{y}^v) \},
\end{equation}
adapted to the spectral decomposition
\begin{equation}
 \label{eq:TMintersK}
T_{\boldsymbol{x}}M=\bigoplus_{a=1}^v \mathcal{V}_{a} (\boldsymbol{x}) \ ,
\end{equation}
whose existence is guaranteed by Theorem \ref{TH}. Let us denote by $m$  the rank of  $(M,\mathcal{H})$.
If  $v\geq m$,
 each operator of the form
\begin{equation}
\label{eq:Ldiag}
\boldsymbol{L}=\sum_{a=1}^{v} \lambda_{a}(\boldsymbol{x})\sum_{j_{a}=1}^{r_{a}}
\frac{\partial}{\partial y^{a, j_{a}}}\otimes \rd y^{a, j_{a}}\ ,
\end{equation}
 is a cyclic Haantjes generator of $\mathscr{H}$, provided that its $v$ eigenvalue fields $\{\lambda_{1}(\boldsymbol{x}), \ldots , \lambda_{v}(\boldsymbol{x})\}$ are  arbitrary but
distinct smooth functions at any point of $U$. In fact, the eigen-distributions of the operator \eqref{eq:Ldiag} are the distributions
$\mathcal{V}_a$;
consequently, by construction they satisfy condition \eqref{eq:cPayne}.  Moreover, this operator also satisfies the assumptions of Proposition \ref{pr:CGK}.
In particular, if the eigenvalues of $\boldsymbol{L}$ are chosen to be
\begin{equation}
 \label{eq:lambdaN}
\lambda_{a}(\boldsymbol{x})=\lambda_a(y^{a,1}, \ldots, y^{a,r_a}  )\qquad\qquad a=1,\ldots, v,
\end{equation}
then $\boldsymbol{L}$ is a cyclic Nijenhuis generator, that is, its Nijenhuis torsion identically vanishes, due to Lemma \ref{lem:LN}.
If $v<m$, a cyclic generator can still be constructed, because we can further decompose some of the distributions $\mathcal{V}_{a}$, as in Eq.~\eqref{eq:Va},  into a direct sum of mutually integrable sub-distributions. Precisely, we have
\[
\mathcal{V}_a=\left \langle  \frac{\partial}{\partial y^{a, 1}},\ldots , \frac{\partial}{\partial y^{a, r_a}} \right  \rangle =\bigoplus_{i_a=1}^{\bar{r}_{i_{a}}} \left \langle   \frac{\partial}{\partial y^{a, 1}},\ldots,\frac{\partial}{\partial y^{a, i_a}}\right  \rangle =\bigoplus_{i_a=1}^{\bar{r}_{i_{a}}} \mathcal{C}_{a, i_a}\ ,
\]
with $\sum \bar{r}_{j_{a}}= r_a$, and the decomposition of $\mathcal{V}_a$ can be realized in such a way that the number of addends appearing into the direct sum
\begin{equation}
 \label{eq:TMintersKC}
T_{\boldsymbol{x}}M=\bigoplus_{a=1,i_a=1}^{v,\bar{r}_{i_{a}}} \mathcal{C}_{a, i_a}
\end{equation}
 again equals $m$.
\end{proof}

\section{Haantjes algebra for the Coulomb--Kepler potential}

In this section, we shall illustrate the novel example of the Haantjes algebra associated with a general class of integrable systems. As a particular instance of this class, we shall discuss a fundamental physical model: the Coulomb--Kepler system in the configuration space given by the affine 3D space $\mathcal{A}_3$.
 The most general Hamiltonian function separable in   spherical coordinates $\{(q_1=r, q_2=\theta,q_3=\phi)\}$ reads
\begin{equation}
 \label{eq:HKepl}
H=\frac{1}{2m}\biggl(p_r+\frac{p_\theta^2}{r^2}+\frac{p^2_\phi}{r^2 \sin^2\theta}
\biggr)+e(r)+\frac{h(\theta)}{r^2}+\frac{s(\phi)}{r^2 \sin^2 \theta} \ ,
\end{equation}
where $e(r), h(\theta), s(\phi)$ are arbitrary smooth functions of their argument.
Let us consider the  semisimple Haantjes algebra of rank three over the phase space  $T^*\mathcal{A}_3$,
generated by the identity operator  $\boldsymbol{I}$, and
\begin{eqnarray}
\label{Kepl1} \boldsymbol{K}_1&=& r^2 \sin^2\theta \biggl( \frac{\partial}{\partial \phi}\otimes \rd \phi+ \frac{\partial}{\partial p_\phi}\otimes \rd p_\phi\biggr) \ , \\
\label{Kepl2} \boldsymbol{K}_2&=& r^2\biggl(
 \frac{\partial}{\partial \theta}\otimes \rd \theta+ \frac{\partial}{\partial p_\theta}\otimes \rd p_\theta+
 \frac{\partial}{\partial \phi}\otimes \rd \phi+
 \frac{\partial}{\partial p_\phi}\otimes \rd p_\phi
 \biggr ) \ .
\end{eqnarray}
This Haantjes algebra allows us to construct  a $\omega \mathscr{H}$ algebra, in the spirit of Ref.~\cite{TT2016prepr}, and in particular the Lenard--Haantjes chain of length three formed by
\[
 \boldsymbol{I}^T \rd H=\rd H\ ,\quad  \boldsymbol{K}_1^T \rd H=\rd I_1\ ,\quad  \boldsymbol{K}_2^T \rd H=\rd I_2 \ .
\]
 It represents a chain of exact one-forms whose potential functions  are  the independent integrals of motion in \emph{involution}
\[
H,\quad I_1=\frac{p_\phi^2}{2m} +s(\phi), \quad I_2=\frac{1}{2m}\biggl(p^2_\theta+\frac{p^2_\phi+s(\phi)}{\sin^2\theta}\biggr)+h(\theta)
\]
for the Hamiltonian system with Hamiltonian function \eqref{eq:HKepl}.

 If  $e(r)=-k/r$ (for a suitable constant $k$),  $h(\theta)=s(\phi)=0$,  one gets the classical Coulomb--Kepler system. For such model, the authors of~\cite{MV} proved that a Nijenhuis recursion operator compatible with
the original symplectic structure can not exist. Indeed, the recursion operator they found provides us with a Lenard chain of integrals of motion which are not independent of the Hamiltonian \eqref{eq:HKepl}.

\section{Cyclic Haantjes operators: Example}

We shall present here an example illustrating some interesting aspects of the Haantjes geometry previously discussed.
In particular, as an application of the main theorems, we develop the procedure of constructing a Haantjes chart  and a cyclic generator for a  semisimple Haantjes algebra of rank $2$, which has been exploited in
the study of   the Jacobi--Calogero Hamiltonian system~\cite{TT2016}.

\subsection{Construction of a Haantjes chart}
 \label{ssec:CG}
Let us consider the affine 3D space $\mathcal{A}_3$, with a cartesian  coordinate system $\{O; \boldsymbol{x}=(x,y,z)\}$ and the two operators
$\boldsymbol{K}_i: TM \rightarrow TM$ defined as
\[
\boldsymbol{K}_1 := y^2 \frac{\partial}{\partial x}\otimes dx +x^2 \frac{\partial}{\partial y} \otimes dy
-xy\left(\frac{\partial}{\partial x} \otimes dy+\frac{\partial}{\partial y} \otimes dx\right) \ ,
\]
\begin{eqnarray*}
\boldsymbol{K}_2:&=&(y^2+z^2) \frac{\partial}{\partial x}\otimes dx +(x^2+z^2) \frac{\partial}{\partial y} \otimes dy+ (x^2+ y^2) \frac{\partial}{\partial z} \otimes dz \\
&-&xy\left(\frac{\partial}{\partial x} \otimes dy+\frac{\partial}{\partial y} \otimes dx\right)
-xz\left(\frac{\partial}{\partial x} \otimes dz+\frac{\partial}{\partial z} \otimes dx\right)
-yz\left(\frac{\partial}{\partial y} \otimes dz+\frac{\partial}{\partial z} \otimes dy\right) \ .
\end{eqnarray*}
They generate an Abelian Haantjes algebra $(\mathcal{A}_3, \mathscr{H})$ of rank $2$, as
\[
\boldsymbol{K}^{2}_{1}=f \boldsymbol{K}_1 \ , \qquad \boldsymbol{K}^{2}_{2} = g \boldsymbol{K}_2 \ , \qquad \boldsymbol{K}_1 \boldsymbol{K}_2 = g \boldsymbol{K}_1 \ ,
\]
where
\[
f(\boldsymbol{x})= x^2+y^2, \qquad g(\boldsymbol{x}) = x^2+y^2+z^2 \ .
\]
Note that $\boldsymbol{K}_2$ is a special case of the inertia tensor \eqref{eq:InT}, for of a single mass point $P_\gamma \equiv O$ with unitary mass, and  $n=3$. Interestingly enough, the two
operators $\boldsymbol{K}_1$ and $\boldsymbol{K}_2$ are the projections onto the  space $\mathcal{A}_3$ of the operators \eqref{Kepl1} and \eqref{Kepl2}, defined over $T^*\mathcal{A}_3$, which are associated with the  model \eqref{eq:HKepl}. The spectra of $\boldsymbol{K}_1$ and
$\boldsymbol{K}_2$ are
\begin{align*}
 Spec(\boldsymbol{K}_1)&=&\{l_1^{(1)}=f, l_2^{(1)}=0\}\\
 Spec(\boldsymbol{K}_2)&=&\{l_1^{(2)}=0, l_2^{(2)}=g \} \ ,
\end{align*}
and their  eigen-distributions read
\begin{align*}
&& \mathcal{D}_1^{(1)}=\langle Y_1\rangle\ ,\quad \mathcal{D}_2^{(1)}=\langle Y_2,Y_3\rangle \ ,
\quad Y_1 := -y\frac{\partial}{\partial x} +x\frac{\partial}{\partial y} \ ,
Y_2 := x\frac{\partial}{\partial x} +y\frac{\partial}{\partial y} \ ,
Y_3 := \frac{\partial}{\partial z} \ ,\\
&& \mathcal{D}_1^{(2)}=\langle Z_1\rangle\ ,\quad \mathcal{D}_2^{(2)}=\langle Z_2,Z_3\rangle \ ,\quad
Z_1 := Y_2+z Y_3 \ ,
Z_2 := Y_1\ ,
Z_3 := -z\frac{\partial}{\partial x} +x\frac{\partial}{\partial z} \ .
\end{align*}
We can construct Haantjes charts for the operators $\boldsymbol{K}_1$ and $\boldsymbol{K}_2$ by determining coordinates adapted to the decomposition \eqref{eq:TVa}.
To this aim, we observe that
\begin{align*}
&& \mathcal{D}_1^{(1)} \cap \mathcal{D}_1^{(2)}=\langle 0\rangle \ ,
\quad
\mathcal{D}_1^{(1)} \cap \mathcal{D}_2^{(2)}=\langle Z_2\rangle \ ,
\quad
\mathcal{D}_2^{(1)} \cap \mathcal{D}_1^{(2)}=\langle Z_1\rangle \ ,\\
&& \mathcal{D}_2^{(1)} \cap \mathcal{D}_2^{(2)}=\langle Z_4\rangle  \ \ ,\quad Z_4=-xz\frac{\partial}{\partial x} -yz\frac{\partial}{\partial y}+(x^2
+y^2) \frac{\partial}{\partial z} \ .
\end{align*}
Therefore, the tangent   spaces can be decomposed as
\begin{equation}
 \label{eq:ExaSD}
T_{\boldsymbol{x}}M= \langle Z_1 \rangle (\boldsymbol{x}) \oplus \langle Z_2\rangle (\boldsymbol{x}) \oplus \langle Z_4\rangle  (\boldsymbol{x}) \ .
\end{equation}
A set of local coordinates adapted  to such decomposition is given (for $ x\neq0$) by the functions
\begin{equation}
\label{eq:coorad}
\begin{split}
x_1&=\frac{y}{x} \quad\Rightarrow \rd x_1\in  {\Bigl(\langle Z_1\rangle \oplus \langle Z_4\rangle\Bigr)}^\circ  \\
x_2&=\sqrt{x^2+y^2+z^2}\quad \Rightarrow  \rd x_2\in {\Bigl(\langle Z_2\rangle \oplus\langle Z_4\rangle \Bigr)}^\circ \\
x_3&= \frac{z}{\sqrt{x^2+y^2+z^2}}\quad \Rightarrow\rd x_3 \in  {\Bigl(\langle Z_1\rangle \oplus \langle Z_2\rangle \Bigr)}^\circ \ .
\end{split}
\end{equation}
Such coordinates  are related to spherical coordinates $(r, \theta, \phi)$ in $\mathcal{A}_3$ as
\[
x_1=\tan\phi \ , \quad x_2=r \ ,\quad x_3=\cos\theta \ .
\]
They are characteristic functions of the spherical web whose fibers are the three foliations generated by:
\begin{itemize}
\item
$ \langle Z_1\rangle \oplus \langle Z_4\rangle$, half-planes issued from  the $z$ axis;
\item
$\langle Z_2\rangle \oplus\langle Z_4\rangle$, spheres centered in $O$;
\item
$\langle Z_1\rangle \oplus \langle Z_2\rangle$ one-folded circular cones with axis $z$ and vertex $O$.
\end{itemize}
In these coordinates, the Haantjes operators $ \boldsymbol{K}_1$ and $ \boldsymbol{K}_2$ take the diagonal form
\begin{align*}
&& \boldsymbol{K}_1=x_2^2(1-x_3^2 )\frac{\partial}{\partial x_1}\otimes dx_1 \\
&& \boldsymbol{K}_2=x_2^2\left(\frac{\partial}{\partial x_1}\otimes dx_1+ \frac{\partial}{\partial x_3}\otimes dx_3\right) \ .
\end{align*}
A cyclic Haantjes generator (admitting three distinct eigenvalues) for the  Haantjes algebra $\mathscr{H}$ is given by
\begin{equation}
\label{eq:LK}
 \boldsymbol{L}= \boldsymbol{K}_1+\boldsymbol{K}_2 \ ,
\end{equation}
whose spectrum is
\[
Spec(\boldsymbol{L})=\{ \lambda_1=-x_2^2(x_3^2-2) ,\, \lambda_2=0,\, \lambda_3= x_2^2\} \ .
\]
The corresponding eigen-distributions are $\langle Z_2\rangle$, $\langle Z_1\rangle$, $\langle Z_4\rangle$, respectively. In fact
\[
 \boldsymbol{K}_1=
 (1+\frac{\lambda_3}{\lambda_1})\,\boldsymbol{L}-\frac{1}{\lambda_1}\,\boldsymbol{L}^2 \ ,
\]
and
\[
 \boldsymbol{K}_2=
 -\frac{\lambda_3}{\lambda_1}\,\boldsymbol{L}+\frac{1}{\lambda_1}\,\boldsymbol{L}^2 \ .
\]
Thus, $\mathscr{H}$ turns out to be a Haantjes subalgebra  of the  cyclic Haantjes algebra of rank three
\begin{equation}
\label{eq:3basis}
\mathcal{L}(\boldsymbol{L})= \left \langle  \boldsymbol{I} , \boldsymbol{L},\boldsymbol{L}^2 \right  \rangle  =\langle \boldsymbol{I},\boldsymbol{K}_1, \boldsymbol{K}_2 \rangle.
\end{equation}

The eigenvalues of $\boldsymbol{L}$ are such  that $\boldsymbol{L}$, in  the original cartesian coordinates $(x,y,z)$,
 takes a simple polynomial form. Precisely, we have
\begin{eqnarray*}
\boldsymbol{L}&=&
\left(2y^2+z^2\right) \frac{\partial}{\partial x}\otimes dx +
\left(2x^2+z^2\right) \frac{\partial}{\partial y}\otimes dy +
\left(x^2+y^2 \right)  \frac{\partial}{\partial z}\otimes dz  \\
&&
-2xy\left(\frac{\partial}{\partial x}\otimes dy+\frac{\partial}{\partial y}\otimes dx \right )
-xz \left(\frac{\partial}{\partial x}\otimes dz+\frac{\partial}{\partial z}\otimes dx \right ) \\
&&
-y z \left(\frac{\partial}{\partial y}\otimes dz+\frac{\partial}{\partial z}\otimes dy \right ) \ .
\end{eqnarray*}
The cyclic Haantjes algebra with basis \eqref{eq:3basis} also admits cyclic Nijenhuis generators, as for instance the operator
\[
\boldsymbol{N}=x_1\frac{\partial}{\partial x_1}\otimes dx_1+x_2^2 \frac{\partial}{\partial x_2}\otimes dx_2+x_3^2 \frac{\partial}{\partial x_3}\otimes dx_3 \ .
\]
However, although it takes a  rational form in the original cartesian coordinates $\{x,y,z\}$, the expression of this operator turns out to be much more complicated than that of $\boldsymbol{L}$. In fact,  the numerators of its components  are polynomials up to   degree 9. This shows that exploiting Haantjes (with non-vanishing Nijenhuis torsion)  cyclic generators can be computationally more convenient than using Nijenhuis cyclic generators.

\section{A comparison with other algebraic structures: Haantjes manifolds and Killing--St\"ackel algebras}

It could be useful to compare our definition of Haantjes algebras over a differentiable manifold, proposed in Section 5, with the notion of \textit{Haantjes manifolds}
recently introduced by Magri~\cite{MGall13}. The main difference between the two constructions
resides in their distinct degree of generality, which obviously reflects the different motivations underlying the two approaches.

In our construction of Haantjes algebras, we are mainly concerned with the abstract, general theory of commuting Haantjes operators, without any reference to additional geometric
structures like exact $1$-forms or symmetry vector fields, that in Magri's theory are essential in order to define  Lenard complexes of commuting vector fields or exact $1$-forms~\cite{MGWDVV,MGall15}.

In other words, although Magri's Haantjes manifolds possess a richer axiomatic structure than our Haantjes algebras, we prefer to articulate our definition on a minimal number of requirements, in order to have a flexible structure,
which in a subsequent step can be made  suitable for the study of  more specific problems, as the separation of variables in the context of integrability or Riemannian geometry. Keeping with this philosophy, we postpone
the introduction of additional geometric structures (as Haantjes chains of exact $1$-forms, symplectic forms~\cite{TT2016prepr}, Poisson bivectors~\cite{TT2017} or Riemannian metrics~\cite{TT2016})
to a further stage of the theory.

Another relevant aspect is that the notion of \textit{Killing--St\"ackel algebra} on an $n$-dimensional Riemannian manifold, due to Benenti et al. ~\cite{BCR02}, can be naturally interpreted in terms of Haantjes algebras of rank $n$.
 In order to compare the two notions, it is useful to
observe that the cyclic  generator \eqref{eq:LK} enjoys a special property:  In fact,
its contravariant form is a  Killing two-tensor  with respect to the Euclidean metric of the affine space $\mathcal{A}_3$. So, we shall call it a \textit{Killing--Haantjes} cyclic generator;
it can be identified with a \emph{characteristic tensor} (CKT) of the \textit{Killing--St\"ackel algebra} \eqref{eq:3basis}.
Thus, the theory of Haantjes cyclic algebras over a Riemannian manifold makes contact with the notion of  Killing--St\"ackel algebras, which offer a geometrical setting for the classical
theory of separation of variables for Hamiltonian systems going back to Eisenhart, St\"ackel, Jacobi, etc.
The main difference between the two algebraic structures is that Killing--St\"ackel
algebras are \emph{vector spaces}  of Killing two-tensors over $\mathbb{R}$. Instead, Haantjes algebras are \emph{modules} over $C^{\infty}(M)$ of Haantjes operators.

At the same time, it is interesting to notice that, starting from a Killing--Haantjes
cyclic generator, one can choose suitable functions to generate other Killing--Haantjes two-tensors, that is, elements of Killing--St\"ackel algebras. The conditions which
such functions must obey are under investigation.

To summarize the comparison among these geometric structures, we can distinguish three different scenarios.

\textit{(i)} When dealing with Killing--St\"ackel algebras,  we are realizing a specific Haantjes algebra with constant coefficients,
which does not always  possess  a cyclic Killing--Haantjes generator. In  fact,  although a CKT does exist, it is not always a cyclic generator, since
to reconstruct the full algebra one would need to combine the powers of the CKT with suitable functions.
Instead, in  Killing--St\"ackel algebras  only linear combinations with constant coefficients are allowed (by definition).

\textit{(ii)} One can consider a generalization of the Killing--St\"ackel  algebra, obtained combining the powers of the CKT by means of functions. In this
case one obtains a larger algebra, which is a full cyclic Haantjes algebra, defined over the ring of smooth functions.

\textit{(iii)} The most general case is that of Haantjes algebras that do not come from  Killing tensors.

\section*{Acknowledgments}
The authors wish to thank F. Magri for having drawn our attention to the theory of Haantjes manifolds and N. Kamran for discussions, a reading of the manuscript and encouragement.

We also thank heartily Y. Kosmann-Schwarzbach for letting us know the  preprint \cite{YKS} before its publication.

 The research of P. T. has been supported by the research project PGC2018-094898-B-I00, Ministerio de Ciencia, Innovaci\'on y Universidades, Spain, and by the ICMAT Severo Ochoa project SEV-2015-0554, Ministerio de Ciencia, Innovaci\'on y Universidades, Spain.
   P. T. and G. T. are members of Gruppo Nazionale di Fisica Matematica (GNFM) of INDAM.




\end{document}